\documentclass[journal,10pt,twocolumn,twoside]{IEEEtran} 
\IEEEoverridecommandlockouts
\usepackage{amsmath}
\usepackage{amsfonts}
\usepackage{cases}
\usepackage{setspace}
\usepackage{fancybox}
\usepackage{subfigure}
\usepackage{epsfig}
\usepackage{graphicx}
\usepackage{epstopdf}
\usepackage{float}
\usepackage{multirow}
\usepackage{color}%
\usepackage{amsmath}
\usepackage{multirow}

\usepackage{indentfirst}
\usepackage{dsfont}
\usepackage{amsfonts}
\usepackage{times,amsmath,color,amssymb,epsfig,cite,subfigure,algorithm,algorithmic}

\newenvironment{proof}[1][Proof]{\begin{trivlist}
\item[\hskip \labelsep {\bfseries #1}]}{\end{trivlist}}

\newcommand{\qed}{\nobreak \ifvmode \relax \else
      \ifdim\lastskip<1.5em \hskip-\lastskip
      \hskip1.5em plus0em minus0.5em \fi \nobreak
      \vrule height0.40em width0.6em depth0.25em\fi}

\newtheorem{lemma}{Lemma}
\newtheorem{theorem}{Theorem}

\begin{document}
\title{Designing Unimodular Sequences with Optimized Auto/cross-correlation properties via Consensus-ADMM/PDMM Approaches}
%
%
%
\author{~Yongchao Wang,~
        ~Jiangtao Wang
        }

%

\maketitle

\begin{abstract}
  Unimodular sequences with good auto/cross-correlation properties are favorable in wireless communication and radar applications.
  In this paper, we focus on designing these kinds of sequences. The main content is as follows: first, we formulate the designing problem as a quartic polynomial minimization problem with constant modulus constraints;
  second, by introducing auxiliary phase variables, the polynomial minimization problem is equivalent to a consensus nonconvex optimization problem; third, to achieve its good approximate solution efficiently, we propose two efficient algorithms based on alternating direction method of multipliers (ADMM) and parallel direction method of multipliers (PDMM);
  fourth, we prove that the consensus-ADMM algorithm can converge to some stationary point of the original nonconvex problem and consensus-PDMM's output is some stationary point of the original nonconvex problem if it is convergent. Moreover, we also analyze the nonconvex optimization model's local optimality and computational complexity of the proposed consensus-ADMM/PDMM approaches.
  Simulation results demonstrate that the proposed ADMM/PDMM approaches outperform state-of-the-art ones in either computational cost or correlation properties of the designed unimodular sequences.
\end{abstract}


\IEEEpeerreviewmaketitle

\section{Introduction}
\IEEEPARstart{U}{\MakeLowercase n}imodular sequences with good auto/cross-correlation properties are very favorable in wireless communication and radar systems. The reasons are twofold: One is they can maximize the amplifier's power efficiency in the transmitter and the other is they can greatly improve the system's performance. For example, when the sequences have low autocorrelation sidelobes, they can improve target detection possibility  \cite{Stoica_07}\cite{Levanon_09}, facilitate synchronization \cite{Spasojevic_01}\cite{Hu_17}\cite{Tseng_00} as well as power control \cite{Schmidt_07}, etc. Moreover, when unimodular sequences have low cross-correlation sidelobe levels, they can be applied to clutter mitigation \cite{Bliss_03}, improving parameter identifiability \cite{Li_07} and distinguishing users \cite{Ponnaluri_07}. Therefore, many researchers are attracted to this field in designing unimodular sequences with good auto/cross-correlation properties.

At the early stage, many studies focused on the autocorrelation property of the considered unimodular sequences. In \cite{Mertens_96}, authors customized an exhaustive search algorithm to construct binary-phase sequences. In \cite{Wang_08}, authors proposed an iterated variable depth searching algorithm to obtain binary-phase sequences with good autocorrelation properties. Besides binary-phase sequences, designing polyphase (or continuous phase) sequences with low sidelobe levels are also investigated widely. In \cite{Borwein_05} and \cite{Nunn_09}, authors proposed two heuristic methods to design polyphase sequences respectively. However, both of them are not capable of designing long sequences due to their high computational complexities. Later, authors in \cite{Stoica_09} and \cite{Stoica_09_letter} proposed two iterative methods named cyclic algorithm-new (CAN) and periodic CAN (PeCAN) respectively to design unimodular aperiodic and periodic sequences. In \cite{Stoica_10}, authors developed a closed-form construction to obtain integrated sidelobe level (ISL) and peak sidelobe level (PSL) lower bounds under a power constraint. In \cite{Soltanalian_12}, authors introduced an algorithm frame framework based on an iterative twisted approximation to design unimodular sequences with a low periodic or aperiodic correlation and zero correlation zone property. In \cite{liang_16}, authors formulated the designing problem as a quartic minimization problem and then customized an alternating direction method of multipliers (ADMM) iteration algorithm to solve it approximately. In \cite{Song_15}--\hspace{-0.001cm}\cite{Zhao_17}, the authors applied the majorization-minimization (MM) technique to minimize autocorrelation sidelobe levels, which can guarantee that its objective function value decreases in every iteration. The authors in \cite{Arriaga_17} designed a strategy of minimizing the generalized weighted ISL measure to obtain the desired unimodular sequences.

In comparison with the above research topic, designing unimodular sequences with both low autocorrelation sidelobe levels and low cross-correlation levels is very challenging. In \cite{He_09}, authors proposed an approach named Weighted Cyclic Algorithms-New (WeCAN) method which can lower cross-correlation levels within certain lag intervals. In \cite{wang_12}, authors formulated the designing problem as a quartic polynomial minimization problem with constant modulus constraints, and then adopted a quasi-Newton solving algorithm to approximate the model's optimal solution. In \cite{Song_16}, authors applied the MM weighted correlation (MM-WeCorr) technique to design these unimodular sequences, which has faster convergence than the WeCAN approach. In \cite{Wu_18}, authors applied the MM technique to design a transmit waveform/receive filter for the MIMO radar with multiple waveform constraints. In \cite{Najafabadi_17}, authors focused on designing sequences with minimum PSL. They formulated the problem of PSL minimization based on Chebyshev distance and exploited the fast-randomized singular value decomposition technique to improve the performance of the proposed algorithm. In \cite{Kerahroodi_17}, authors considered both the continuous and discrete phase constraints and proposed a coordinate-descent method to design low sidelobe sequences. In \cite{Li_18}, authors formulated the ISL and weighted ISL minimization problems as quartic polynomial optimization models, and then simplified them into quadratic problems via the MM technique.

In this paper, we focus on designing unimodular sequences with optimized autocorrelation sidelobe levels and cross-correlation levels via  consensus-ADMM/PDMM approaches. First, the designing problem is formulated as a quartic polynomial minimization problem with constant modulus constraints. Then, we introduce auxiliary phase variables to the polynomial minimization problem and reformulate it as a consensus nonconvex optimization problem. Moreover, we propose two efficient solving algorithms, based on ADMM \cite{Goldstein_14}--\hspace{-0.001cm}\cite{Zhang_18} and parallel direction method of multipliers (PDMM) techniques \cite{ADMM_con}--\hspace{-0.001cm}\cite{Wang_14}, to efficiently achieve the problem's solution. Finally, we show several analyses on the proposed consensus-ADMM/PDMM algorithms, such as convergence, local optimality, and efficient implementations. Simulation results demonstrate the effectiveness of the proposed approaches.

The remaining sections of the paper are organized as follows. In Section II, the problem's formulation procedure is presented. Two solving algorithms named consensus-ADMM and consensus-PDMM as well as their performance analyses are presented in Section III and Section IV respectively. Finally, Section V presents some numerical results, and the conclusions are given in Section VI.

\emph{Notation}: Bold lowercase and uppercase letters denote column vectors and matrices and italics denote scalars. $\mathbb{R}$ and $\mathbb{C}$ denote the real field and complex field respectively. The superscripts $(\cdot)^*$, $(\cdot)^T$ and  $(\cdot)^H$ denote conjugate, transpose and conjugate transpose respectively.  $|\cdot|$ denotes the absolute value. The subscripts $\|\cdot\|_2$ and $\|\cdot\|_F$ denote Euclidean vector norm and Frobenius matrix norm. $\nabla(\cdot)$ represents the function's gradient. ${\rm Re}(\cdot)$ takes the real part of the complex variable and ${\rm Tr}(\cdot)$ denotes the trace of a matrix. ${\rm mat}(\cdot,N,M)$ reshapes a vector to an $N\times M$ matrix. $\langle \mathbf{x},\mathbf{y}\rangle$ denotes the dot product of $\mathbf{x}$ and $\mathbf{y}$. ${\rm vec}(\cdot)$ vectorizes a matrix by stacking its columns on top of one another.

\section{Problem Formulation}
\label{sec:format}

Consider a set of $M$ unimodular sequences $\{\mathbf{x}_m\}_{m=1}^M$ and the length of each sequence is $N$, i.e.,  $\mathbf{x}_m=[x_{1,m}, \cdots ,x_{N,m}]^T$ and $|x_{i,m}|=1$. The correlation function of sequences $\mathbf{x}_i$ and $\mathbf{x}_j$ at lag $n$ is defined as
\begin{equation}\label{correlation}
\begin{split}
r_{ijn}=&\sum_{k=n+1}^Nx_{k,i}^*x_{k-n,j}=\mathbf{x}_i^H\mathbf{S}_{n}\mathbf{x}_j,\\
&i,j=1, \cdots ,M; n=-N+1, \cdots ,N-1.
\end{split}
\end{equation}
Here, $\mathbf{S}_n$ is defined as an off-line diagonal 0-1 matrix. When $n>0$, nonzero elements located in the upper off-line of the matrix are shown in \eqref{Sn}.
\begin{equation}\label{Sn}
\begin{split}
 &\qquad\qquad \quad n~{\rm zeros}\\
&{\bf{S}}_n=
\left[
  \begin{array}{cccccc}
    ~ & \overbrace{0~ \cdots ~0} & 1 & ~ & \mathbf{\scalebox{1.5}0} & ~ \\
    ~ & ~ & ~ &  \ddots &~ & ~  \\
    ~ & ~ & ~ & ~ & 1 & ~ \\
    ~ & \mathbf{\scalebox{2.5}0} & ~ & ~ & ~ & ~  \\
  \end{array}
\right].
\end{split}
\end{equation}
When $n<0$, nonzero elements are located in the lower off-line of the matrix. Specifics, since ${\mathbf{S}}_n = {\mathbf{S}}_{-n}^H$, there exists $r_{ijn}=r_{ij-n}^*$.

For sequences $\{\mathbf{x}_m\}_{m=1}^M$, we define set $\mathcal{T}$ corresponding to the lag interval of interest. Then, the autocorrelation metric, called an integrated sidelobe level (ISL), can be written as
\begin{equation}\label{ISL}
{\rm ISL} = \sum_{i=1}^M\sum_{n\in\mathcal{T}\backslash 0}|r_{iin}|^2,
\end{equation}
and the cross-correlation metric, called a cross-correlation level (CCL), can be written as
\begin{equation}\label{CCL}
{\rm CCL} = \sum_{i=1}^M\sum_{\substack{j=1\\j\neq i}}^M\sum_{n\in\mathcal{T}}|r_{ijn}|^2.
\end{equation}
Moreover, we define the correlation matrix at lag $n$
 \[
\begin{split}
{\mathbf{R}}_n = &\left[
\begin{array}{cccc}
  r_{11n} & r_{12n} & \cdots & r_{1Mn} \\
  r_{21n} & r_{22n} & \cdots & r_{2Mn} \\
  \vdots & ~ & \ddots & \vdots \\
  r_{M1n} & \cdots & \cdots & r_{MMn}
\end{array}
\right].
\end{split}
\]
 Since the sequences $\{\mathbf{x}_m\}_{m=1}^M$ can be denoted by the $N$-by-$M$ matrix, i.e., $\mathbf{X}=[\mathbf{x}_1, \cdots ,\mathbf{x}_m]$, then $\mathbf{R}_n$ can be obtained through
 \begin{equation}\label{Rn_X}
{\mathbf{R}}_n = \mathbf{X}^H{\bf{S}}_n\mathbf{X}.
\end{equation}
Then, combining \eqref{ISL}, \eqref{CCL}, and \eqref{Rn_X}, we have
\begin{equation}\label{ISL_CCL}
{\rm ISL} + {\rm CCL}
  =
\|{\mathbf{X}}^H{\mathbf{X}}-N\mathbf{I}\|_F^2+\sum_{n\in\mathcal{T}} \|\mathbf{X}^H{\bf{S}}_n\mathbf{X}\|_F^2,
\end{equation}
where ${\mathbf{I}}$ is the identity matrix.

Thus, a compact optimization model for designing unimodular sequences with minimized ISL/CCL can be formulated as
\begin{subequations}\label{original model}
\begin{align}
&\min_{\mathbf{X} \in \mathbb{C}^{N \times M}}\ \ \ \|{\mathbf{X}}^H{\mathbf{X}}-N\mathbf{I}\|_F^2+\sum_{n\in \mathcal{T}}
\|\mathbf{X}^H{\mathbf{S}}_n\mathbf{X}\|_F^2, \label{original model a}\\
&{\rm subject\ to}\ |x_{i,m}|=1,\ i=1,\cdots,N, m=1,\cdots,M. \label{original model b}
\end{align}
\end{subequations}

Solving model \eqref{original model} directly is difficult since the objective function \eqref{original model a} is a fourth-order polynomial and the constraints are constant modulus equalities. However, since every element in $\mathbf{X}$ is a constant modulus, i.e., $x_{i,m}=e^{j\phi_{i,m}}$, we drop constant modulus constraints and formulate problem \eqref{original model} to the following minimization problem
\begin{equation}\label{unconstrained model}
\begin{split}
&\hspace{0.45cm} \min_{\mathbf{\Phi} } \hspace{0.65cm} \sum_{n\in \mathcal{T}} f_n(\mathbf{\Phi}),\\
& {\rm subject\ to} \hspace{0.3cm} 0\preceq\mathbf{\Phi}\prec2\pi.
\end{split}
\end{equation}
where
\begin{equation}\label{fn}
f_n(\mathbf{\Phi})=\|\mathbf{X}(\mathbf{\Phi})^H{\bf{S}}_n\mathbf{X}(\mathbf{\Phi})-N\mathbf{I}\delta_n\|_F^2,\ \ n\in \mathcal{T},
\end{equation}
 the constraint $0\preceq\mathbf{\Phi}\prec2\pi$ means all the elements $\phi_{i,m}$ in $\mathbf{\Phi}$ belong to $[0,2\pi)$, and $\delta_n$ in \eqref{fn} denotes the Dirac-$\delta$ function. Problem \eqref{unconstrained model} can further be equivalent to the following consensus-like problem \eqref{Consensus} by introducing a set of auxiliary variables $\{\mathbf{\Phi}_n, n\in \mathcal{T}\}$.
\begin{equation}\label{Consensus}
\begin{split}
&\hspace{0.25cm} \min_{\mathbf{\Phi},\{\mathbf{\Phi}_n\} } \hspace{0.4cm} \sum_{n\in \mathcal{T}} f_n(\mathbf{\Phi}_n)\\
& {\rm subject\ to} \hspace{0.3cm} \mathbf{\Phi}_n = \mathbf{\Phi},\ 0\preceq\mathbf{\Phi}\prec2\pi,\ n\in \mathcal{T}.
\end{split}
\end{equation}
In comparison with \eqref{unconstrained model}, the major benefit of the consensus-like problem \eqref{Consensus} is the flexibility of allowing $f_n(\mathbf{\Phi}_n)$ to handle its local variable independently. Sequentially, we will design two efficient algorithms, named consensus-ADMM and consensus-PDMM, to solve \eqref{Consensus} approximately, but efficiently. Moreover, we show several analyses on the proposed algorithms related to  convergence, local optimality, and efficient implementations.

\section{Customized Consensus-ADMM/PDMM Solving Algorithms}
\label{sec:pagestyle}

\subsection{Consensus-ADMM Algorithm Framework}

The augmented Lagrangian function of problem \eqref{Consensus} can be written as
\begin{equation}\label{AL ADMM}
\begin{split}
&\mathcal{L}(\mathbf{\Phi},\{\mathbf{\Phi}_n,\mathbf{\Lambda}_n,n\in \mathcal{T}\}) \\
=& \sum_{n\in \mathcal{T} } \left(f_n({\mathbf{\Phi}_n})
 + \langle\mathbf{\Lambda}_n,\mathbf{\Phi}_n-\mathbf{\Phi}\rangle+ \frac{\rho_n}{2}\|\mathbf{\Phi}_n-\mathbf{\Phi}\|_F^2\right),
\end{split}
\end{equation}
where $\mathbf{\Lambda}_n$ and $\rho_n$, $n\in \mathcal{T}$, are Lagrangian multipliers and penalty parameters respectively. We further define
\begin{equation}\label{Ln ADMM}
\begin{split}
&\mathcal{L}_n(\mathbf{\Phi},\mathbf{\Phi}_n,\mathbf{\Lambda}_n) \\
=& f_n({\mathbf{\Phi}_n})
 + \langle\mathbf{\Lambda}_n,\mathbf{\Phi}_n-\mathbf{\Phi}\rangle+ \frac{\rho_n}{2}\|\mathbf{\Phi}_n-\mathbf{\Phi}\|_F^2,
\end{split}
\end{equation}
where $n\in \mathcal{T}$.

Then, the consensus-ADMM algorithm framework for solving problem \eqref{Consensus} can be written as
\begin{subequations}\label{ori_ADMM}
\begin{align}
&\mathbf{\Phi}^{k+1} = \underset{0\preceq\mathbf{\Phi}\prec2\pi} {\arg \min}\ \ \mathcal{L}\left(\mathbf{\Phi},\{\mathbf{\Phi}_n^k,\mathbf{\Lambda}_n^k,n\in \mathcal{T}\}\right),\label{step1 ori ADMM}\\
&\mathbf{\Phi}_n^{k+1} =  \underset{\mathbf{\Phi}_n} {\arg \min}\ \ \mathcal{L}_n\left(\mathbf{\Phi}^{k+1},\mathbf{\Phi}_n,\mathbf{\Lambda}_n^k\right),n\in \mathcal{T},\label{step2 ori ADMM}\\
&\mathbf{\Lambda}_n^{k+1} = \mathbf{\Lambda}_n^{k} + \rho_n(\mathbf{\Phi}_n^{k+1} -\mathbf{\Phi}^{k+1}),n\in \mathcal{T}.\label{step3 ori ADMM}
\end{align}
\end{subequations}
where $k$ is iteration number.

{\it Remarks:} First, for different $n\in\mathcal{T}$, the variables in \eqref{step2 ori ADMM} and \eqref{step3 ori ADMM} are independent of each other. It means that the $|\mathcal{T}|$ paired problems \eqref{step2 ori ADMM} and \eqref{step3 ori ADMM} can be implemented in parallel, where $|\mathcal{T}|$ is set $\mathcal{T}$'s size.
Second, the main difficulty of implementing the consensus-ADMM algorithm \eqref{ori_ADMM} lies in how to solve problem \eqref{step2 ori ADMM} since functions $\mathcal{L}_n\left(\mathbf{\Phi}^{k}, \mathbf{\Phi}_n,\mathbf{\Lambda}_n^k\right)$ are nonconvex related to variables $\mathbf{\Phi}_n$. However, the following lemma indicates that $\{f_n(\mathbf{\Phi}), \ n\in \mathcal{T}\}$ are continuous, differentiable with respect to the phase variable and have Lipschitz continuous gradients (see detailed proof in Appendix A).

\begin{lemma}\label{Lipschtiz continuous}
  Gradients $\{\nabla f_n(\mathbf{\Phi}), n\in\mathcal{T}\}$ are Lipschitz continuous with constants $L_n$, i.e.,
\begin{equation}\label{Lipschitiz}
\|\nabla f_n({\mathbf{\Phi}})-\nabla f_n({\mathbf{\hat{\Phi}}})\|_F\leq L_n\|{\mathbf{\Phi}}-{\mathbf{\hat{\Phi}}}\|_F,\ n\in \mathcal{T},
\end{equation}
where
\begin{equation}\label{Ln}
   L_n > 4(M-1)(N+1).
\end{equation}
\end{lemma}

Based on Lemma \ref{Lipschtiz continuous}, we have the following inequality
\begin{equation}\label{upperbound Un}
\begin{split}
&\mathcal{L}_n\left(\mathbf{\Phi}^{k}, \mathbf{\Phi}_n,\mathbf{\Lambda}_n^k\right)\leq f_n({\mathbf{\Phi}}^{k})+\langle \nabla f_n({\mathbf{\Phi}}^{k}),{\mathbf{\Phi}_n}-{\mathbf{\Phi}}^{k}\rangle \\
&\hspace{1.2cm}+ \langle\mathbf{\Lambda}_n^k,\mathbf{\Phi}_n-\mathbf{\Phi}^{k}\rangle+ \frac{\rho_n+L_n}{2}\|\mathbf{\Phi}_n-\mathbf{\Phi}^{k}\|_F^2.
\end{split}
\end{equation}
Let right hand side of inequality \eqref{upperbound Un} be $\mathcal{U}_n\left(\mathbf{\Phi}^{k+1},\mathbf{\Phi}_n,\mathbf{\Lambda}_n^k\right)$. Then, we customize the following consensus-ADMM solving algorithm
\begin{subequations}\label{ADMM}
\begin{align}
&\mathbf{\Phi}^{k+1} = \underset{0\preceq\mathbf{\Phi}\prec2\pi} {\arg \min}\ \ \mathcal{L}\left(\mathbf{\Phi},\{\mathbf{\Phi}_n^k,\mathbf{\Lambda}_n^k,n\in \mathcal{T}\}\right),\label{step1 ADMM}\\
&\mathbf{\Phi}_n^{k+1} = \underset{\mathbf{\Phi}_n} {\arg \min}\  \ \mathcal{U}_n\left(\mathbf{\Phi}^{k+1},\mathbf{\Phi}_n,\mathbf{\Lambda}_n^k\right), n\in \mathcal{T},\label{step2 ADMM}\\
&\mathbf{\Lambda}_n^{k+1} = \mathbf{\Lambda}_n^{k} + \rho_n(\mathbf{\Phi}_n^{k+1} -\mathbf{\Phi}^{k+1}), n\in \mathcal{T}.\label{step3 ADMM}
\end{align}
\end{subequations}

Since $\mathcal{L}\left(\mathbf{\Phi},\{\mathbf{\Phi}_n^k,\mathbf{\Lambda}_n^k,n\in \mathcal{T}\}\right)$ and $\mathcal{U}_n(\mathbf{\Phi}^{k+1}, \mathbf{\Phi}_n,\mathbf{\Lambda}_n^k)$ are strongly convex quadratic functions with respect to $\mathbf{\Phi}$ and $\mathbf{\Phi}_n$, optimal solutions of problems \eqref{step1 ADMM} and \eqref{step2 ADMM} can be obtained by solving linear equations \eqref{linear equations a} and \eqref{linear equations b} respectively.
\begin{subequations}\label{linear equations}
\begin{align}
 &\nabla_{\mathbf{\Phi}}\mathcal{L}\left(\mathbf{\Phi},\{\mathbf{\Phi}_n^k,\mathbf{\Lambda}_n^k,n\in \mathcal{T}\}\right)=0,\label{linear equations a} \\ &\nabla_{\mathbf{\Phi}_n}\mathcal{U}_n\left(\mathbf{\Phi}^{k+1},\mathbf{\Phi}_n,\mathbf{\Lambda}_n^k\right)=0.\label{linear equations b}
\end{align}
\end{subequations}
Then, we project the solutions onto the feasible region and obtain
\begin{subequations}\label{solutions}
\begin{align}
&\mathbf\Phi^{k+1}\!=\!\underset{[0,2\pi)}{\Pi} \left(\frac{1}{|\mathcal{T}|}\sum\limits_{n\in \mathcal{T}}\left(\mathbf{\Phi}_n^k+\frac{\mathbf{\Lambda}_n^k}{\rho_n}\right)\right), \label{solution phi}\\
&\mathbf\Phi_n^{k+1} = \mathbf\Phi^{k+1} -\frac{\nabla f_n(\mathbf\Phi^{k+1})+\mathbf{\Lambda}_n^k}{\rho_n+L_n},n\in \mathcal{T}. \label{solution phi n}
\end{align}
\end{subequations}
Combining \eqref{step3 ADMM} and \eqref{solutions}, we summarize the customized consensus-ADMM algorithm in Table \ref{table ADMM}.

\begin{table}[tbp]
\renewcommand \arraystretch{1.2}
\caption{The customized consensus-ADMM algorithm }
\label{table ADMM}
\centering
\begin{tabular}{l}
 \hline\hline
 {\bf Initialization:} Compute Lipschitz constants $\{L_n,\!n\!\in\!\mathcal{{T}}\}$  \\
 \hspace{0.2cm} according to \eqref{Ln}.
  Set iteration index $k\!=\!1$, initialize \\
 \hspace{0.2cm} $\mathbf{\Phi}^1$ and $\{\mathbf\Lambda_n^1, n\in\mathcal{T}\}$ randomly, and let $\{\mathbf{\Phi}^1= \mathbf{\Phi}_n^1,$ \\
 \hspace{0.2cm} $ n\in \mathcal{T}\}$. \\
  {\bf repeat} \\
  \hspace{0.2cm} S.1 Compute $\mathbf\Phi^{k+1}$ via \eqref{solution phi}, i.e., \\
  \hspace{0.9cm} $\mathbf\Phi^{k+1}\!=\!\underset{[0,2\pi)}{\Pi} \left(\frac{1}{|\mathcal{T}|}\sum\limits_{n\in \mathcal{T}}\left(\mathbf{\Phi}_n^k+\frac{\mathbf{\Lambda}_n^k}{\rho_n}\right)\right)$.\\
  \hspace{0.2cm} S.2 Compute $\{\mathbf\Phi_n^{k+1}, n\in \mathcal{{T}}\}$ via \eqref{solution phi n} in parallel, i.e.,\\
  \hspace{0.9cm}
  $
   \mathbf{\Phi}_n^{k+1} = \mathbf\Phi^{k+1} -\frac{\nabla f_n(\mathbf\Phi^{k+1})+\mathbf{\Lambda}_n^k}{\rho_n+L_n}.
  $ \\
  \hspace{0.2cm} S.3 Compute $\{\mathbf{\Lambda}_n^{k+1}, n\in \mathcal{{T}}\}$ via \eqref{step3 ADMM} in parallel, i.e., \\
  \hspace{0.9cm} $\mathbf{\Lambda}_n^{k+1}\! = \!\mathbf{\Lambda}_n^{k} + \rho_n(\mathbf{\Phi}_n^{k+1} -\mathbf{\Phi}^{k+1})$.\\
 {\bf until} some preset termination criterion is satisfied.\\
 \hspace{0.75cm} Let $\mathbf{\Phi}^{k+1}$ be the output.\\
 \hline\hline
\end{tabular}
\end{table}

\subsection{Consensus-PDMM Algorithm Framework}
In this subsection, we develop a consensus-PDMM algorithm with a full parallel implementation structure to solve problem \eqref{original model}. In it, the updated process during one iteration can be executed in one phase, which could provide a flexible asynchronous updated manner that is more suitable for some real applications.

\begin{table}[tbp]
\renewcommand \arraystretch{1.2}
\caption{The proposed consensus-PDMM algorithm }
\label{table PDMM}
\centering
\begin{tabular}{l}
 \hline\hline
 {\bf Initialization:} Set $M$ and $N$. Compute $L_n$ according to \\ \hspace{0.2cm} \eqref{Ln}. Set iteration index $k=1$, choose $\mathbf{\Phi}^1$ and $\mathbf\Lambda_n^1$ \\
 \hspace{0.2cm} randomly and let $\{\mathbf{\Phi}^1 = \mathbf{\Phi}_n^1, n\in \mathcal{T}\backslash 0\}$. \\
  {\bf repeat} \\
  \hspace{0.2cm} S.1 Compute $\mathbf\Phi^{k+1}$ via \eqref{solu phi}, i.e., \\
  \hspace{0.9cm} $\mathbf{{\Phi}}^{k+1} \!=\!\underset{[0,2\pi)}{\Pi}\!\left(\frac{L_0\mathbf{\Phi}^k-\nabla f_0(\mathbf{\Phi}^k)+\!\!\sum\limits_{n\in \mathcal{T}\backslash 0}\!\!\!\left(\mathbf{\Lambda}_n^k\!+\!\rho_n\mathbf{\Phi}_n^k\right)}{L_0+\sum\limits_{n\in \mathcal{T}\backslash 0}\rho_n}\right)$.\\
  \hspace{0.2cm} S.2 Compute $\{\mathbf\Phi_n^{k+1}\!, n\in \mathcal{{T}}\backslash 0\}$ via \eqref{solu phi n} in parallel, i.e.,\\
  \hspace{0.9cm} $\displaystyle{
   \mathbf{{\Phi}}_n^{k+1} \!=\!\frac{L_n\mathbf{\Phi}_n^{k}+\rho_n \mathbf{\Phi}^{k}-\mathbf{\Lambda}_n^k-\nabla f_n(\mathbf{\Phi}_n^{k})}{L_n+\rho_n}. }
  $   \\
  \hspace{0.2cm} S.3 Compute $\{\mathbf{\Lambda}_n^{k+1}\!, n\in \mathcal{{T}}\backslash 0\}$ via \eqref{step3 PDMM} in parallel, i.e., \\
  \hspace{0.9cm} $\mathbf{\Lambda}_n^{k+1}\! = \!\mathbf{\Lambda}_n^{k} + \rho_n(\mathbf{\Phi}_n^{k+1} -\mathbf{\Phi}^{k})$.\\
 {\bf until} some preset termination criterion is satisfied.\\
 \hspace{0.75cm} Let $\mathbf{\Phi}^{k+1}$ be the output.\\
 \hline\hline
\end{tabular}
\end{table}

Specifically, consensus problem \eqref{Consensus} can be equivalent to
\begin{equation}\label{Consensus PDMM}
\begin{split}
&\hspace{0.3cm}\min_{\mathbf{\Phi},\left\{\mathbf{\Phi}_n\right\} } \hspace{0.4cm} f_0(\mathbf{\Phi})+\sum_{n\in \mathcal{T}\backslash 0} f_n(\mathbf{\Phi}_n) \\
& {\rm subject\ to}\hspace{0.3cm}0\preceq\mathbf{\Phi}\prec2\pi,\mathbf{\Phi}_n = \mathbf{\Phi},\ \ n\in \mathcal{T}\backslash 0,
\end{split}
\end{equation}
Its augmented Lagrangian function can also be written as (see \eqref{Ln ADMM})
\begin{equation}\label{AL PDMM}
\begin{split}
\hspace{-0.3cm}\mathcal{L}\!\left(\!\mathbf{\Phi},\!\{\mathbf{\Phi}_n,\mathbf{\Lambda}_n,\!n\!\in \!\mathcal{T}\backslash 0\}\!\right)\! =\! f_0(\mathbf{\Phi}) \!+\!\!\!
 \!\sum_{n\in \mathcal{T}\backslash 0}\!\!  \mathcal{L}_n\!\left(\mathbf{\Phi},\mathbf{\Phi}_n,\mathbf{\Lambda}_n\right).
\end{split}
\end{equation}
Then, the proposed consensus-PDMM algorithm\footnotemark can be described as
\footnotetext{Here, we should note that the proposed consensus-PDMM algorithm is different to the parallel methods in\cite{Deng_14}\cite{Wang_14}, which focus on the minimization of block-separable convex functions subject to linear constraints.}
\begin{subequations}\label{PDMM_ori}
\begin{align}
&\mathbf{\Phi}^{k+1} = \underset{0\preceq\mathbf{\Phi}\prec2\pi} {\arg \min}\ \ \mathcal{L}\left(\mathbf{\Phi},\{\mathbf{\Phi}_n^{k},\mathbf{\Lambda}_n^{k},n\in \mathcal{T}\backslash 0\}\right), \label{step1 PDMM ori}\\
&\mathbf{\Phi}_n^{k+1} =  \underset{\mathbf{\Phi}_n} {\arg \min}\ \  \mathcal{L}_n\left(\mathbf{\Phi}^{k},\mathbf{\Phi}_n,\mathbf{\Lambda}_n^k\right), \ \  n\in \mathcal{T}\backslash 0, \label{step2 PDMM ori}\\
&\mathbf{\Lambda}_n^{k+1} = \mathbf{\Lambda}_n^{k} + \rho_n(\mathbf{\Phi}_n^{k+1} -\mathbf{\Phi}^{k}),\ \ n\in \mathcal{T}\backslash 0. \label{step3 PDMM ori}
\end{align}
\end{subequations}
One can see that $\mathbf\Phi^k$ (not $\mathbf\Phi^{k+1}$) and $\mathbf\Lambda_n^k$ are involved in solving \eqref{step2 PDMM ori}. This fact admits problems \eqref{step1 PDMM ori} and \eqref{step2 PDMM ori} can be solved in parallel.
According to Lemma \ref{Lipschtiz continuous}, we can obtain upper-bound functions $\mathcal{U}(\mathbf{\Phi}, \{\mathbf{\Phi}_n^k,\mathbf{\Lambda}_n^k,n\in\mathcal{T}\backslash 0\})$ and $\mathcal{U}_n\left(\mathbf{\Phi}^{k}, \mathbf{\Phi}_n,\mathbf{\Lambda}_n^k\right)$ of $\mathcal{L}\left(\mathbf{\Phi},\{\mathbf{\Phi}_n^{k},\mathbf{\Lambda}_n^{k},n\in \mathcal{T}\backslash 0\}\right)$ and $\mathcal{L}_n\left(\mathbf{\Phi}^{k},\mathbf{\Phi}_n,\mathbf{\Lambda}_n^k\right)$ respectively in the following
 \begin{equation}\label{upperbound U0}
\begin{split}
 & \hspace{-0.2cm} \mathcal{U}(\!\mathbf{\Phi}\!, \{\mathbf{\Phi}_n^k,\!\mathbf{\Lambda}_n^k,\! n\!\in\!\mathcal{T}\backslash \!0\}\!)
\!=\!f_0({\mathbf{\Phi}}^{k})\!+\! \langle \nabla f_0( {\mathbf{\Phi}}^{k}),\mathbf\Phi\!-\!\mathbf\Phi^k\rangle \\ &\hspace{1.5cm}+\frac{L}{2}\|{\mathbf{\Phi}}-{\mathbf{\Phi}}^{k}\|_F^2 + \sum_{n\in \mathcal{T}\backslash 0}\mathcal{L}_n(\mathbf{\Phi},\mathbf{\Phi}_n^k,\mathbf{\Lambda}_n^k).
\end{split}
\end{equation}
\begin{equation}\label{upp Un}
\begin{split}
&\hspace{-0.3cm}\mathcal{U}_n\left(\mathbf{\Phi}^{k}, \mathbf{\Phi}_n,\mathbf{\Lambda}_n^k\right)\!=\! f_n({\mathbf{\Phi}}_n^{k})\!+\!\langle \nabla f_n({\mathbf{\Phi}}_n^{k}),{\mathbf{\Phi}_n}\!-\!{\mathbf{\Phi}}_n^{k}\rangle \\
&\hspace{-0.3cm}+\frac{L_n}{2}\|\mathbf{\Phi}_n\!-\!\mathbf{\Phi}_n^{k}\|_F^2\!+\! \langle\mathbf{\Lambda}_n^k,\mathbf{\Phi}_n\!-\!\mathbf{\Phi}^{k}\rangle+\!\frac{\rho_n}{2}\|\mathbf{\Phi}_n\!-\!\mathbf{\Phi}^{k}\|_F^2.
\end{split}
\end{equation}
Then, instead of minimizing nonconvex functions $\mathcal{L}\left(\mathbf{\Phi},\{\mathbf{\Phi}_n^{k},\mathbf{\Lambda}_n^{k},n\in \mathcal{T}\backslash 0\}\right)$ and $\mathcal{L}_n\left(\mathbf{\Phi}^{k},\mathbf{\Phi}_n,\mathbf{\Lambda}_n^k\right)$ directly, \eqref{PDMM_ori} can be relaxed to
\begin{subequations}\label{PDMM_relax}
\begin{align}
&\mathbf{{\Phi}}^{k+1} = \underset{0\preceq\mathbf{\Phi}\prec2\pi}{\arg\min}\ \mathcal{U}\left(\mathbf{\Phi},\{\mathbf{\Phi}_n^k,\mathbf{\Lambda}_n^k,n\in \mathcal{T}\backslash 0\}\right),\label{step1 PDMM}\\
&\mathbf{{\Phi}}_n^{k+1} = \underset{\mathbf{\Phi}_n}{\arg\min}\ \mathcal{U}_n\left(\mathbf{\Phi}^{k},\mathbf{\Phi}_n,\mathbf{\Lambda}_n^k\right), n\in \mathcal{T}\backslash 0,\label{step2 PDMM}\\
&\mathbf{{\Lambda}}_n^{k+1} = \mathbf{\Lambda}_n^{k} + \rho_n(\mathbf{{\Phi}}_n^{k+1} -\mathbf{\Phi}^{k}), n\in \mathcal{T}\backslash 0.\label{step3 PDMM}
\end{align}
\end{subequations}
Since $\mathcal{U}(\mathbf{\Phi}, \{\mathbf{\Phi}_n^k,\mathbf{\Lambda}_n^k,n\in\mathcal{T}\backslash 0\})$ and $\mathcal{U}_n\left(\mathbf{\Phi}^{k}, \mathbf{\Phi}_n,\mathbf{\Lambda}_n^k\right)$ are strongly quadratic, optimal solutions of problems \eqref{step1 PDMM} and \eqref{step2 PDMM} can be obtained easily by setting their gradients to zero, solving the linear equations and projecting the solutions onto the corresponding feasible regions, which lead to
\begin{subequations}\label{solutions PDMM}
\begin{align}
&\mathbf{{\Phi}}^{k+1} \!\!=\!\!\underset{[0,2\pi)}{\Pi}\!\!\left(\!\!\frac{L_0\mathbf{\Phi}^k\!\!-\!\!\nabla f_0(\mathbf{\Phi}^k)\!+\!\!\!\!\sum\limits_{n\in \mathcal{T}\backslash 0}\!\!\!\left(\mathbf{\Lambda}_n^k\!+\!\rho_n\mathbf{\Phi}_n^k\right)}{L_0+\sum\limits_{n\in \mathcal{T}\backslash 0}\rho_n}\!\right)\!,\label{solu phi}\\
&\mathbf{{\Phi}}_n^{k+1} \!=\!\frac{L_n\mathbf{\Phi}_n^{k}+\rho_n \mathbf{\Phi}^{k}-\mathbf{\Lambda}_n^k-\nabla f_n(\mathbf{\Phi}_n^{k})}{L_n+\rho_n}, n\in\mathcal{T}\backslash 0. \label{solu phi n}
\end{align}
\end{subequations}

In Table \ref{table PDMM}, we summarize the proposed consensus-PDMM algorithm.

\begin{figure*}
\begin{equation}\label{gradient f}
\begin{split}
\nabla f_n({\mathbf{\Phi}}) =
\left\{
\begin{array}{l}
  {\rm mat}\! \!\left(2{\rm Re}\left(\left(\frac{\partial {{\bf v}_0({\mathbf{\Phi}})}}{\partial \phi_{1,1}}\cdots\!\frac{\partial {{\bf v}_0}({\mathbf{\Phi}})}{\partial \phi_{N,M}}\right)^{H}\left({{\bf v}_0}({\mathbf{\Phi}}) - {\bf{c}}\right)\right),N,M\right), \ \ n = 0, \\
  {\rm mat} \left(2{\rm Re}\left(\left(\frac{\partial{\bf{v}}_n(\mathbf{\Phi})}{\partial \phi_{1,1}}\cdots\frac{\partial{\bf{v}}_n(\mathbf{\Phi})}{\partial \phi_{N,M}}\right)^{H}{\bf{v}}_n({\mathbf{\Phi}})\right), N,M\right),\hspace{0.8cm} n\in \mathcal{T}\backslash 0.
\end{array}
\right.
\end{split}
\end{equation}
\end{figure*}

\begin{figure*}[!htbp]
\normalsize
\begin{equation}\label{XX}
\begin{split}
\frac{\partial \mathbf{X(\Phi)}^H\mathbf{X(\Phi)}}{\phi_{i,m}}\!\!=\!\!
\left[
  \begin{array}{ccccccc}
    ~ &  ~& ~ & je^{j(\!\phi_{i,m}\!-\phi_{i,1}\!)} & ~ & ~ & ~ \\
    ~ &  \! \! \mathbf{\scalebox{3.5}0}\!\! & ~& \vdots & ~  & \!\! \mathbf{\scalebox{3.5}0}\!\!  & ~ \\
   ~  &~ &~ & je^{j(\!\phi_{i,m}\!-\phi_{i,m\!-\!1}\!)}&~ &~  \\
   \! \!-\!je^{j(\!\phi_{i,1}\!-\!\phi_{i,m}\!)} &\!\!\cdots\!\! &\!-\!je^{j(\phi_{i,m-1}\!-\phi_{i,m}\!)} & 0 & \!-\!je^{j(\!\phi_{i,m+1}\!-\phi_{i,m}\!)}&  \!\!\cdots\!\! & \!-\!je^{j(\!\phi_{i,M}\!-\phi_{i,m}\!)}\!\!  \\
    ~ & ~ &~ &je^{j(\phi_{i,m}\!-\phi_{i,m+1}\!)}  &~ & ~& ~  \\
    ~ &  \!\!\mathbf{\scalebox{3.5}0} \!\! & ~ & \vdots & ~ & \!\! \mathbf{\scalebox{3.5}0} \!\! & ~  \\
    ~ & ~ &~ & je^{j(\!\phi_{i,m}\!-\phi_{i,M}\!)} &~ & ~& ~  \\
  \end{array}
\right]
\end{split}
\end{equation}
\hrulefill
\vspace*{4pt}
\end{figure*}

\section{Performance Analysis}
\subsection{ Convergence Issue}
 Before showing the convergence theorem, we present Lemmas \ref{upperbound function}-\ref{lemma lower bound} and their proofs in Appendix B. Based on these lemmas, we show that if proper parameters are chosen, the augmented Lagrangian function $\mathcal{L}\left(\cdot\right)$  is {\it sufficient descent} in every iteration and is also lower-bounded, which leads $\mathcal{L}\left(\cdot\right)$ to convergence as $k\rightarrow+\infty$.

 Then, we have Theorem \ref{theorem admm} to show the convergence properties of the proposed consensus-ADMM algorithm (the proof is given in Appendix \ref{proof theorem 1}).
\begin{theorem}\label{theorem admm}
 Let $\left(\mathbf{\Phi}^k,\{\mathbf{\Phi}_n^k,\mathbf{\Lambda}_n^k,n\in\mathcal{{T}})\}\right)$ be the sequence generated by the proposed consensus-ADMM algorithms.
 If penalty parameters $\rho_n$ and Lipschitz constants $L_n$ satisfy $\rho_n\geq9L_n$, we have the following convergence results
\begin{equation}\label{conv variables}
\begin{split}
&\lim\limits_{k\rightarrow+\infty}\mathbf{\Phi}^{k}=\mathbf{\Phi}^{*},\ \ \lim\limits_{k\rightarrow+\infty}\mathbf{\Phi}_n^{k}=\mathbf{\Phi}_n^{*}, \\
&\lim\limits_{k\rightarrow+\infty}\mathbf{\Lambda}_n^{k}=\mathbf{\Lambda}_n^{*},\ \ \ \ \ \mathbf{\Phi}^{*}=\mathbf{\Phi}_n^{*}.
\end{split}
\end{equation}
Moreover, $\mathbf\Phi^*$ is a stationary point of problem \eqref{unconstrained model}, i.e., it satisfies the following inequality
\begin{equation}\label{stationary point}
\left\langle \sum_{n\in \mathcal{T}}\nabla f_n({\mathbf{\Phi}^{*}}),\mathbf{\Phi}-\mathbf{\Phi}^{*}\right\rangle\geq0,\ 0\preceq\mathbf{\Phi}\prec2\pi. \end{equation}

\end{theorem}

\begin{theorem}\label{theorem pdmm}
 In the consensus-PDMM algorithm, the penalty parameters $\rho_n$ and Lipschitz constants $L_n$  are set to satisfy $\rho_n\geq9L_n$. Let $\left(\mathbf{\Phi}^k,\{\mathbf{\Phi}_n^k,\mathbf{\Lambda}_n^k,n\in\mathcal{{T}})\}\right)$ be the sequence generated by the proposed consensus-PDMM algorithms. If \eqref{conv variables} holds, where $n\in\mathcal{T}\backslash 0$, then limit point $\mathbf\Phi^*$ is a stationary point of problem \eqref{unconstrained model}.
\end{theorem}

{\it Remarks:}  The proof of Theorem 2 is presented in Appendix \ref{proof theorem pdmm}.
Here, we should strengthen that Theorem \ref{theorem pdmm} just states the quality of the limit point when the consensus-PDMM algorithm is convergent.
 To date, the convergence analysis of the PDMM algorithm for the general nonconvex optimization model is still an open problem.
 Some state-of-the-art results on this topic, such as \cite{Deng_14}--\hspace{-0.001cm}\cite{Hong_17}, cannot be followed since the nonconvex model \eqref{unconstrained model} cannot satisfy their specific conditions.
 However, the simulation results presented in the next section show that the proposed consensus-PDMM algorithm always converges, and the generated unimodular sequences have good correlation levels.

\subsection{Local Optimality}
Theorems \ref{theorem admm}-\ref{theorem pdmm} show convergence properties of the proposed consensus-ADMM/PDMM algorithms. In this subsection, we presnet a theoretical bound on the quality of local minima of the model \eqref{unconstrained model} (see proof in Appendix \ref{proof local opt}).

\begin{theorem}\label{theorem local opt}
Let $f(\mathbf{\Phi})=\sum\limits_{n\in{\mathcal{T}}} f_n(\mathbf{\Phi})$.
Then, any local minimizer $\mathbf\Phi^*$ of problem \eqref{unconstrained model} is a $\frac{1}{2}$-approximation of its global minimum, i.e.,
\begin{equation}\label{local opt f}
\begin{split}
\frac{f(\mathbf\Phi^*)-f_{\rm min}}{f_{\rm max}-f_{\rm min}}\leq\frac{1}{2},\ \ {0\preceq\mathbf{\Phi}^*\prec2\pi},
\end{split}
\end{equation}
where $f_{\rm min}$ and $f_{\rm max}$ are the global minimum and global maximum value of the objective function in \eqref{unconstrained model} respectively.
\end{theorem}

\subsection{Efficient Implementations}

Observing the proposed consensus-ADMM/PDMM algorithms in Table I and Table II, we can see that the main computational difficulty lies in calculating $\{\nabla f_n(\mathbf{\Phi}), n\in\mathcal{T}\}$. In the following, we show that the gradients can be obtained efficiently by exploiting their special sparsity structures.

First, we define vectors ${\bf{v}}_n = {\rm vec}(\mathbf{X}(\mathbf{\Phi})^H{\bf{S}}_n\mathbf{X}(\mathbf{\Phi})), n\in\mathcal{{T}}$. Then, $f_n(\mathbf{\Phi})$ can be simplified as
\begin{equation}\label{eq_f}
\begin{split}
f_n(\mathbf{\Phi}) = \left\{
                       \begin{array}{ll}
                         \|{ {\bf v}_0}({\mathbf{\Phi}}) - {\bf{c}}\|_2^2 ,& n=0,  \\
                         \|{\bf{v}}_n(\mathbf{\Phi})\|_2^2, & n\in \mathcal{T}\backslash 0,
                       \end{array}
                     \right.
\end{split}
\end{equation}where ${\bf{c}}={\rm vec}(N{\bf{I}})$.
Then, we can compute $\nabla f_n({\mathbf{\Phi}})$ as \eqref{gradient f}.
Second, from \eqref{XX}, we can see that there are $2(M\!-\!1)$ nonzero elements in $\frac{\partial \mathbf{X(\Phi)}^H\mathbf{X(\Phi)}}{\partial\phi_{i,m}}$. It indicates that $2(M-1)$ complex multiplication operations at most are needed to obtain $\frac{\partial{{\bf v}_0^H(\mathbf{\Phi})}}{\partial\phi_{i,m}}\!{({\bf v}_0(\mathbf{\Phi})\!-\!\bf{c})}$ and $\frac{\partial{{\bf{v}}_{n}^H(\mathbf{\Phi})}}{\partial\phi_{i,m}}{\bf{v}}_{n}(\mathbf{\Phi})$. Since ${\mathbf{\Phi}}$ is an $N$-by-$M$ matrix, computational cost of obtaining all gradients $\{\nabla f_n({\mathbf{\Phi}}),n\in\mathcal{T}\}$ in each iteration is roughly  $\mathcal{O}(2M^2N|\mathcal{T}|)$.
Furthermore, observing \eqref{step3 ADMM} and \eqref{solutions}, we can see that the computational cost of other terms is far less than$\nabla f_n({\mathbf{\Phi}}),n\in\mathcal{T}$. This is also true for \eqref{step3 PDMM} and \eqref{solutions PDMM}. Hence, we can conclude that the total cost in each consensus-ADMM/PDMM iteration is roughly $\mathcal{O}(2M^2N|\mathcal{T}|)$.


\section{Simulation results}
\label{sec:typestyle}

\begin{figure*}[htbp]
\centering
\subfigure{
\begin{minipage}[t]{0.5\linewidth}
\centering
\centerline{\includegraphics[scale=0.53]{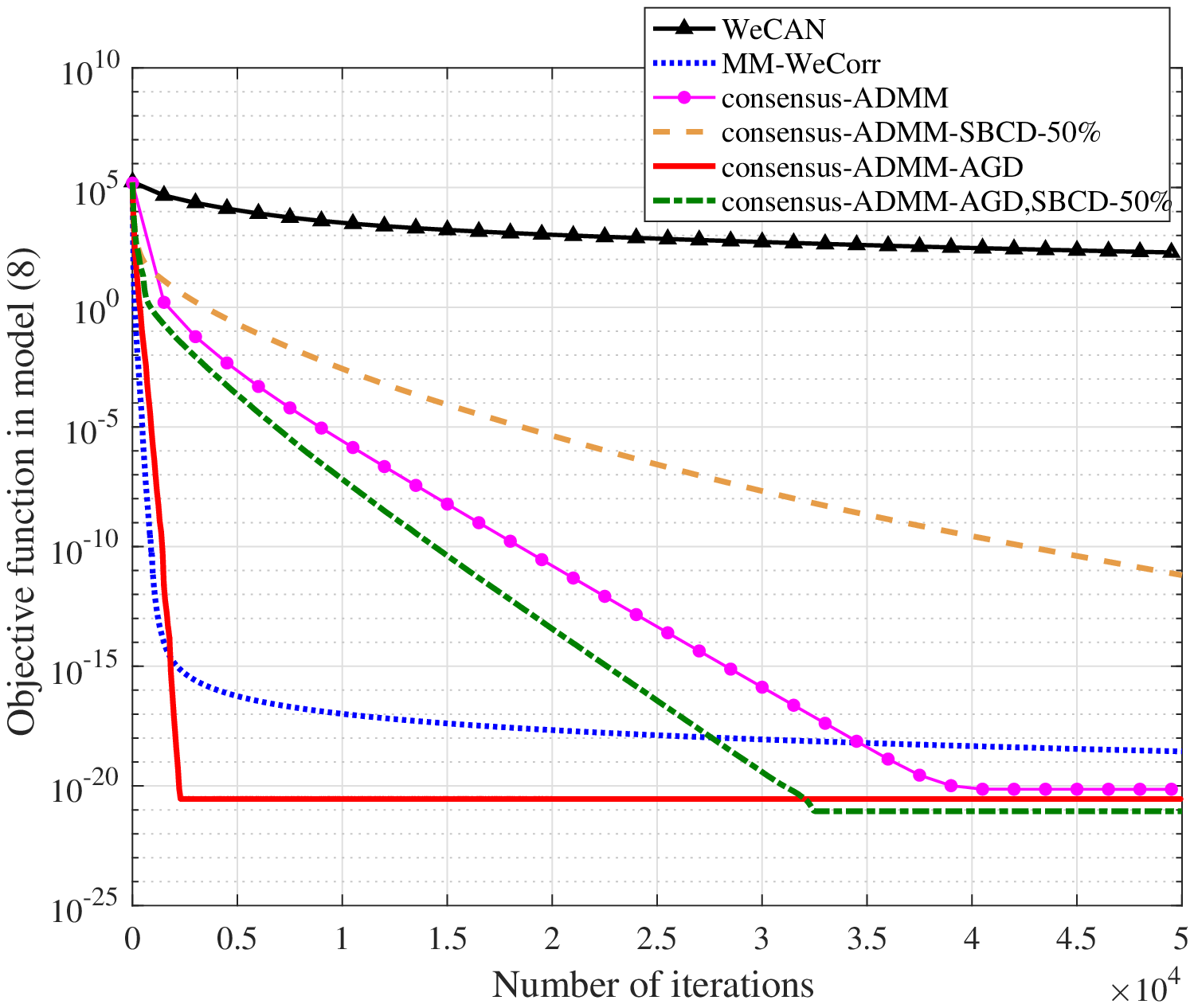}}
\end{minipage}%
}%
\subfigure{
\begin{minipage}[t]{0.5\linewidth}
\centering
\centerline{\includegraphics[scale=0.53]{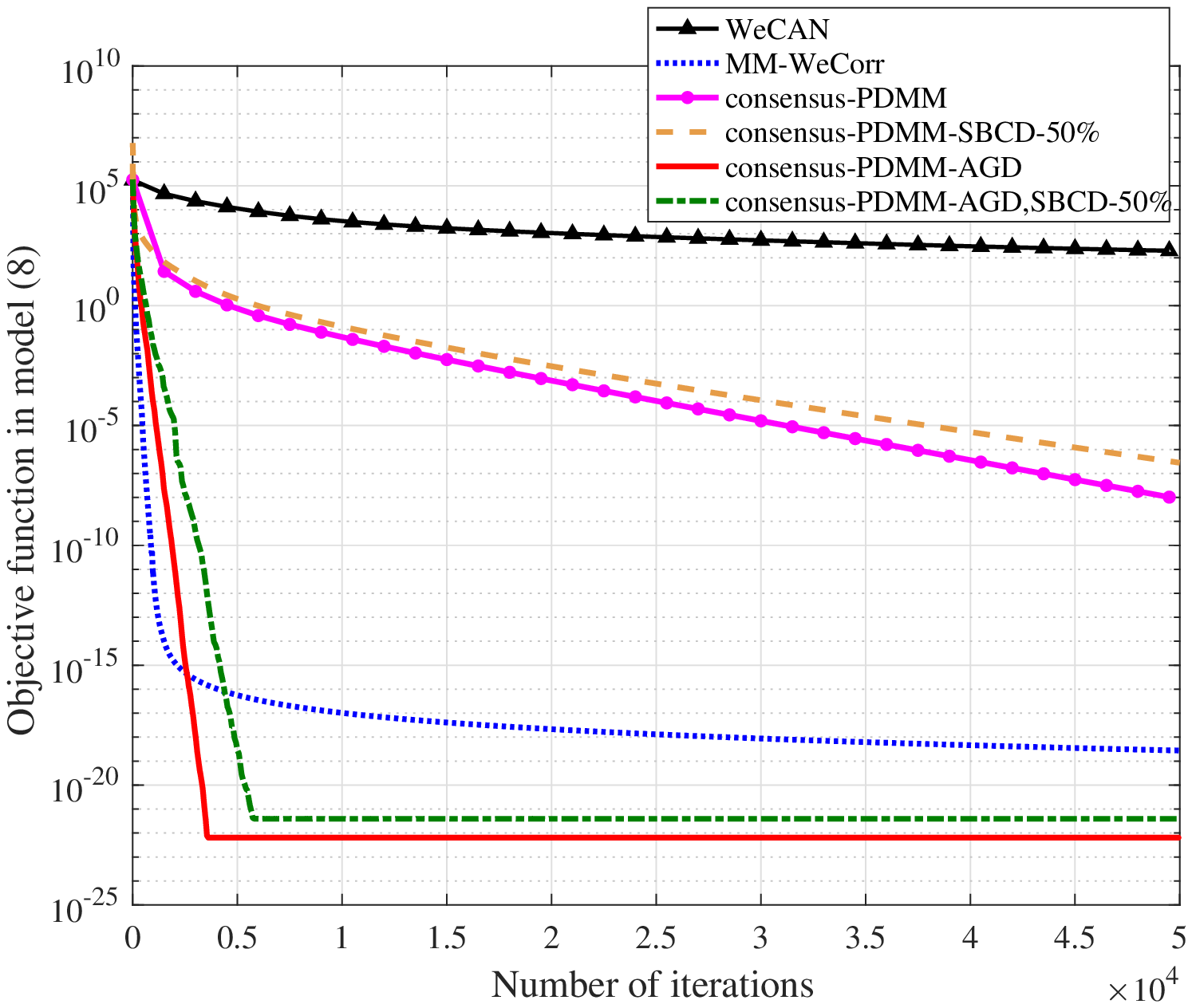}}
\end{minipage}
}
\caption{Comparisons of convergence performance with $N=256,M=3,\mathcal{T}=[1,39]$. SBCD-$50\%$ means that half of the elements in set $\mathcal{T}$ are updated.}
\label{conv fig 256}
\end{figure*}

\begin{figure*}[htbp]
\centering
\subfigure{
\begin{minipage}[t]{0.5\linewidth}
\centering
\centerline{\includegraphics[scale=0.55]{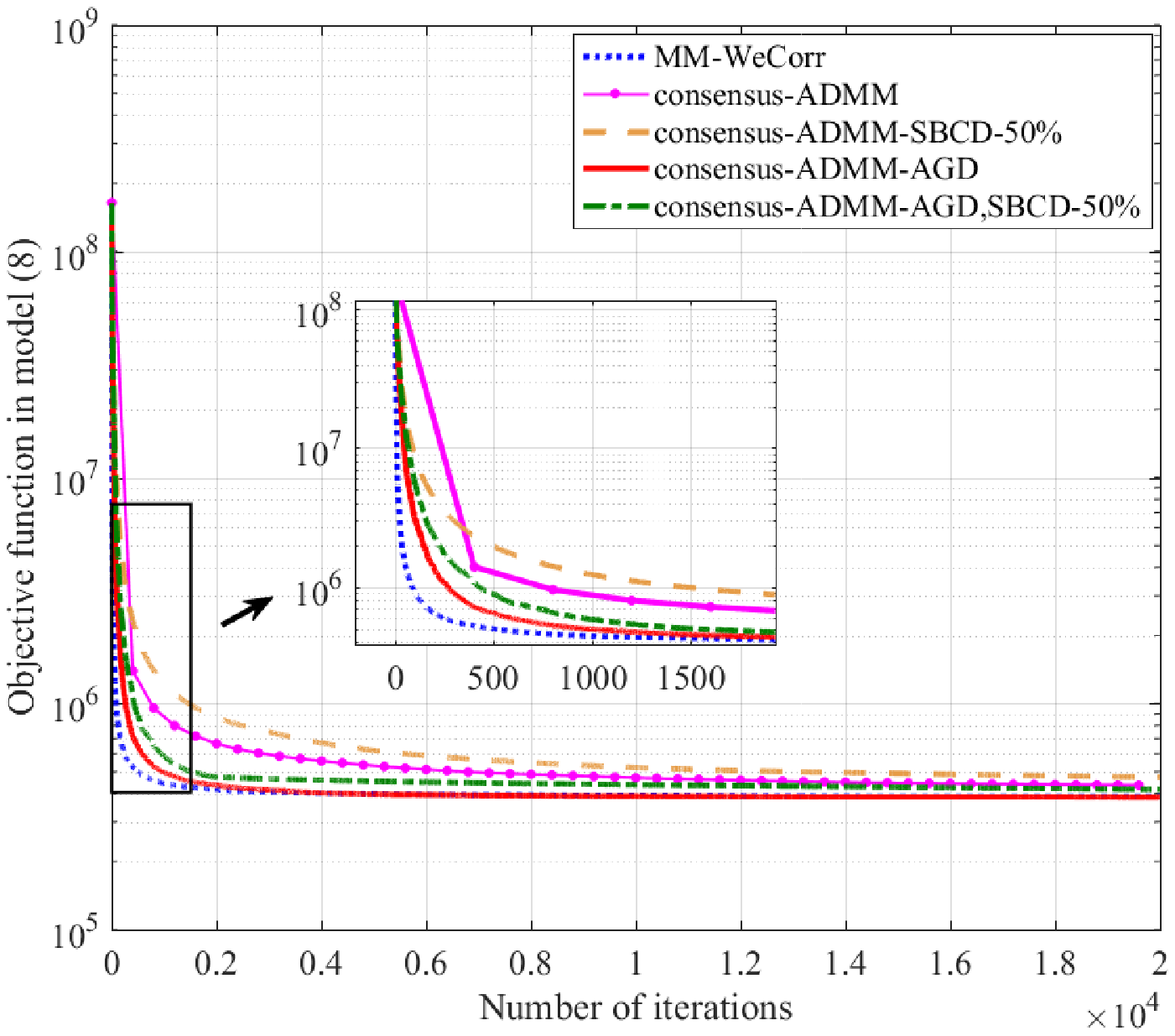}}
\end{minipage}%
}%
\subfigure{
\begin{minipage}[t]{0.5\linewidth}
\centering
\centerline{\includegraphics[scale=0.55]{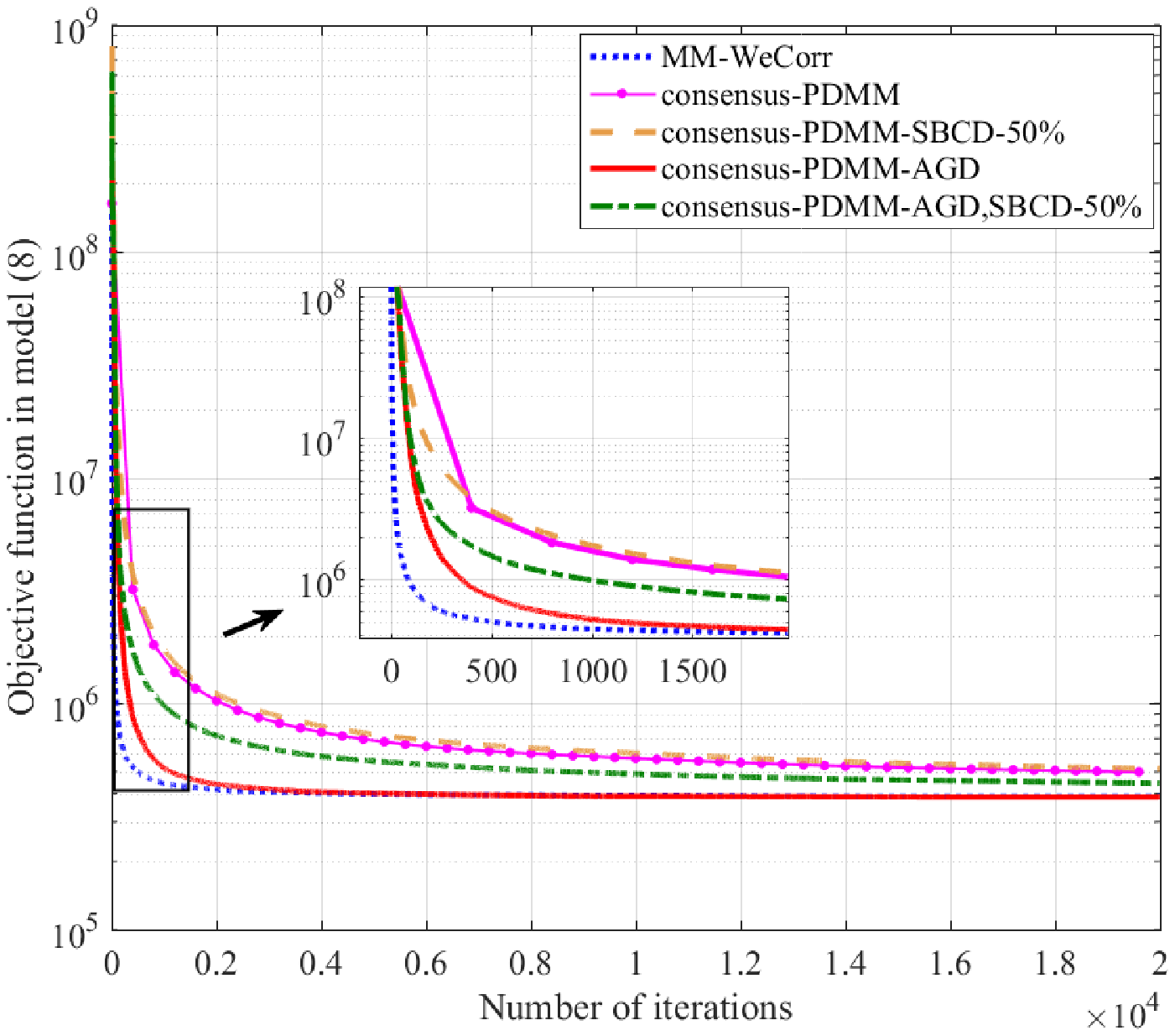}}
\end{minipage}
}
\caption{Comparisons of convergence performance with $N=2048,M=32,\mathcal{T}=[1,39]$.}
\label{conv fig 2048}
\end{figure*}

In this section, several numerical examples are presented to show the performance of the proposed consensus-ADMM/PDMM algorithms.
The simulation parameters are set as follows: For the consensus-ADMM algorithm, we define primal/dual residuals \cite{Boyd_11} at the $k$-th iteration as
\[
\begin{split}
\mathbf{R}^{k}_n = \rho_n(\mathbf{\Phi}_n^{k+1} -\mathbf{\Phi}^{k}),\ \mathbf{D}^{k}=\mathbf{\Phi}^{k+1}-\mathbf{\Phi}^{k}.
\end{split}
\]
For the consensus-PDMM algorithm, $n\in\mathcal{T}\backslash 0$. Then, the termination criterion in Table I and Table II is set as
$
 \sum\limits_{n\in \mathcal{T} }\left\|\mathbf{R}^{k}_n\right\|^2_F +|\mathcal{T}|\left\|\mathbf{D}^{k}\right\|^2_F\leq \epsilon,
$
or the maximum iteration number is reached. In the simulations, we set $\epsilon=10^{-4}$ and maximum iteration number as $5\times10^4$. 
Moreover, to improve the algorithms' performance, we exploited stochastic block coordinate descent (SBCD) and accelerated gradient descent (AGD) \cite{Wang_19} to reduce the computational complexity and speed up convergence respectively.
 In comparison, two state-of-the-art methods, WeCAN \cite{He_09} and MM-WeCorr \cite{Song_16}, are carried out here. All approaches
are initialized with the random phase sequence. Besides, all experiments are performed in MATLAB 2016b/Windows 7 environment on a computer with 2.1GHz Intel 4100$\times$2 CPU and 64GB RAM.

Figures \ref{conv fig 256}-\ref{conv fig 2048} show the convergence characteristics of the proposed consensus-ADMM/PDMM algorithms and other comparison algorithms. Here, it should be noted that we do not give the exact proof of the convergence for the consensus-PDMM algorithm. However, from these figures, we can see that all the algorithms show pretty converge results. Specifically, We-Can converges slowest and MM-WeCorr enjoys pretty fast converge speed. Moreover, AGD strategy can speed up the convergence of our proposed ADMM/PDMM approaches very well. In comparison, SBCD strategy slows down the convergence rate. However, we should note that it has lower computational complexity. The parameter of $50\%$\footnotemark can be changed to attain a tradeoff between convergence rate and computational complexity.

\footnotetext{In the $k$-th iteration, elements are chosen from $\mathcal{T}$ to construct its subset $\mathcal{N}^{k}$ with the probability
$
{\rm Pr}(n\in\mathcal{N}^{k}) = p_n$. In the simulations, $p_n$ is set as $50\%$.
See details in \cite{Tseng_01}.
}

\begin{figure*}[htbp]
\centering
\subfigure{
\begin{minipage}[t]{0.5\linewidth}
\centering
\centerline{\includegraphics[scale=0.41]{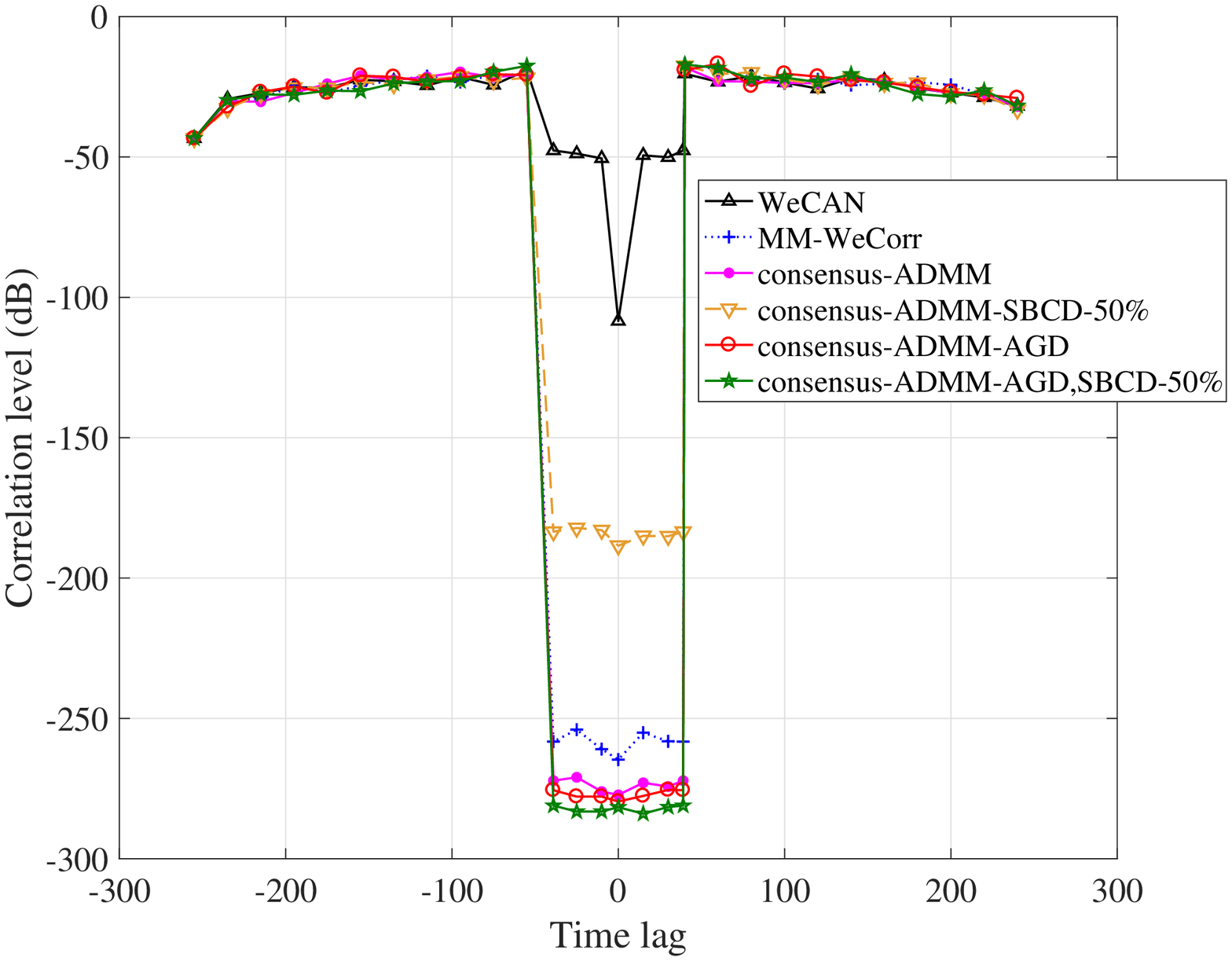}}
\end{minipage}%
}%
\subfigure{
\begin{minipage}[t]{0.5\linewidth}
\centering
\centerline{\includegraphics[scale=0.41]{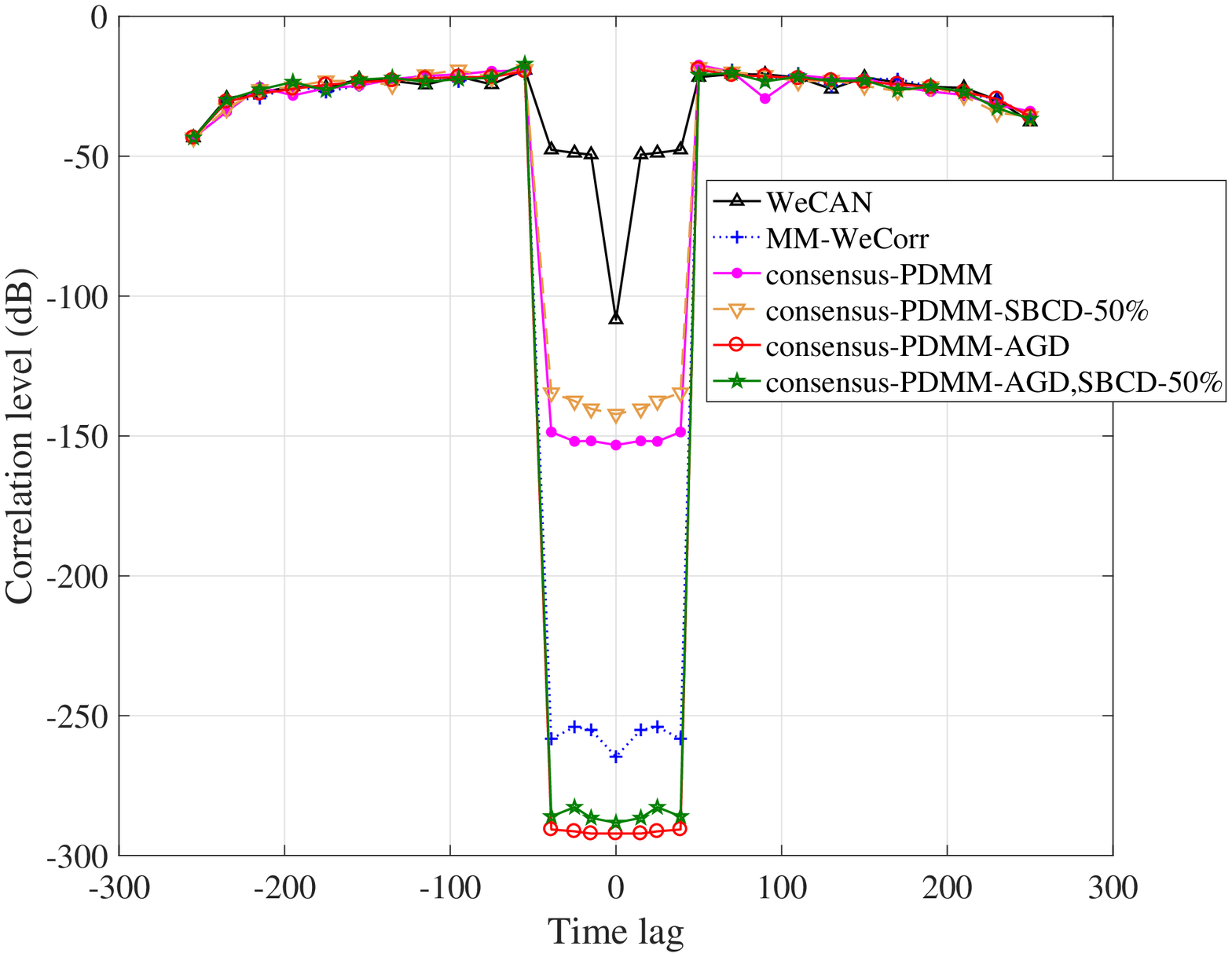}}
\end{minipage}
}
\caption{Correlation levels with $N=256,M=3,\mathcal{T}=[0,39]$.}
\label{corr 1-39}
\end{figure*}

\begin{figure*}[htbp]
\centering
\subfigure{
\begin{minipage}[t]{0.5\linewidth}
\centering
\centerline{\includegraphics[scale=0.315]{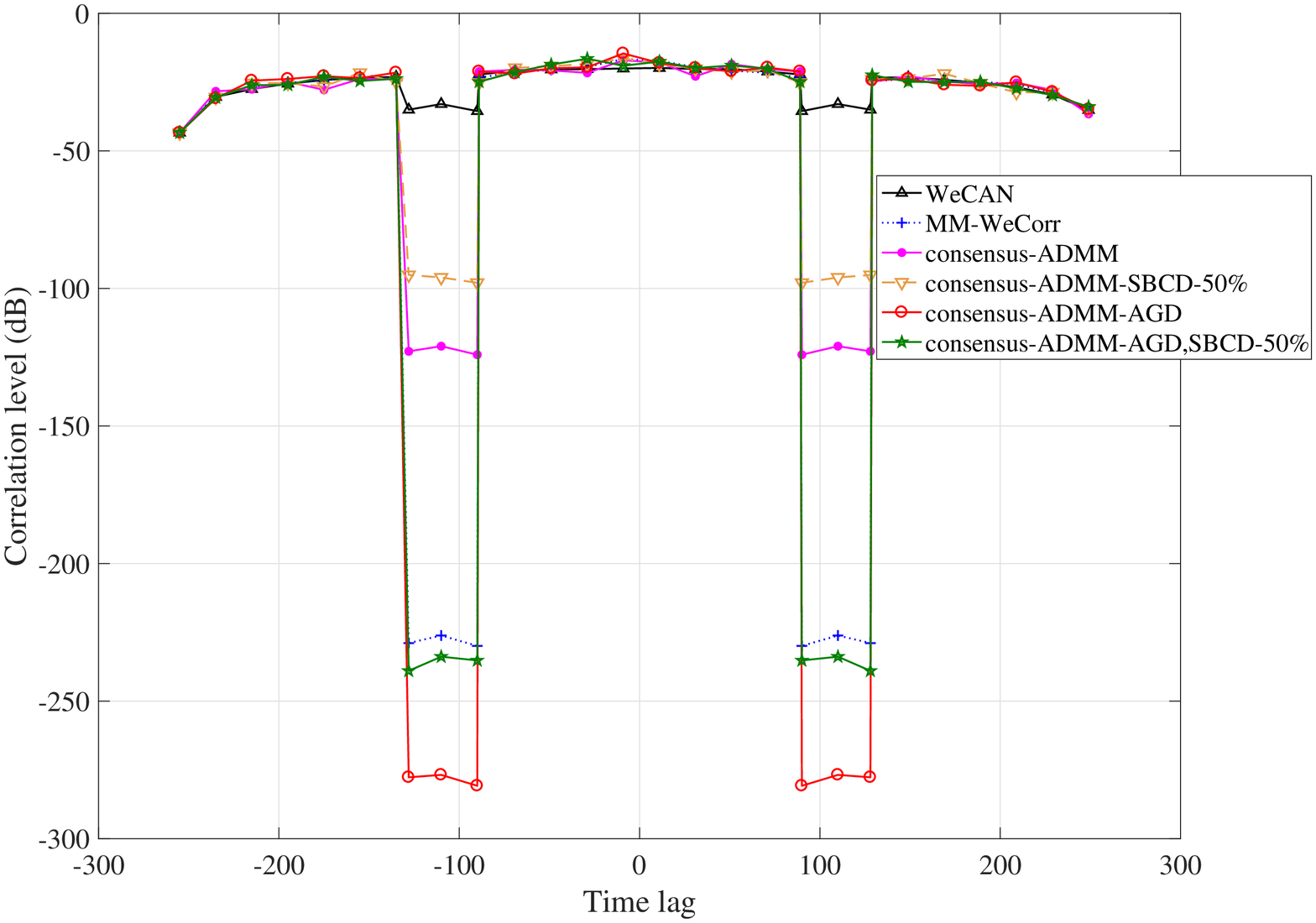}}
\end{minipage}%
}%
\subfigure{
\begin{minipage}[t]{0.5\linewidth}
\centering
\centerline{\includegraphics[scale=0.315]{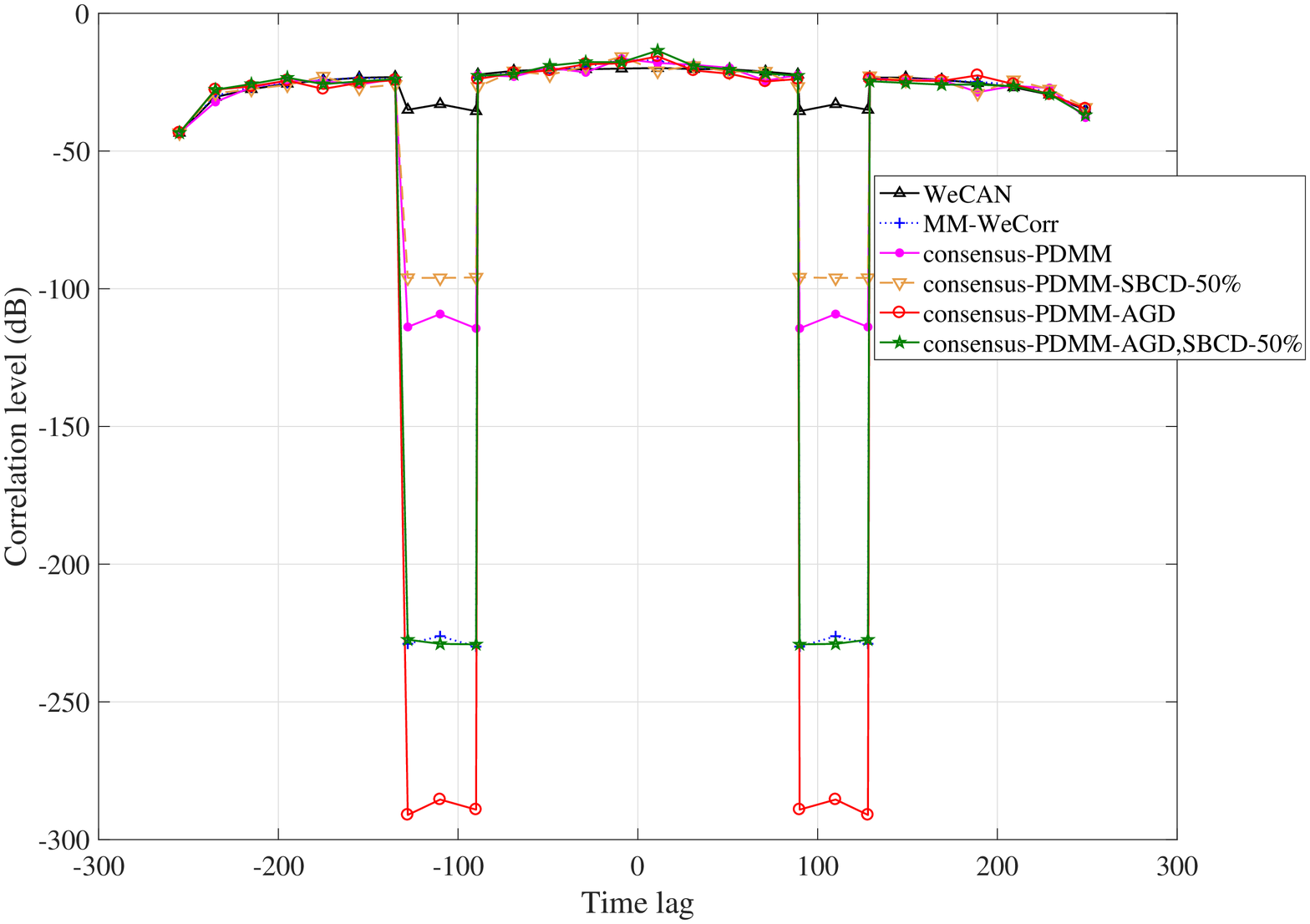}}
\end{minipage}
}
\caption{Correlation levels with $N=256,M=3,\mathcal{T}=[90,128]$.}
\label{corr 90-128}
\end{figure*}

\begin{figure*}[htbp]
\centering
\subfigure{
\begin{minipage}[t]{0.5\linewidth}
\centering
\centerline{\includegraphics[scale=0.37]{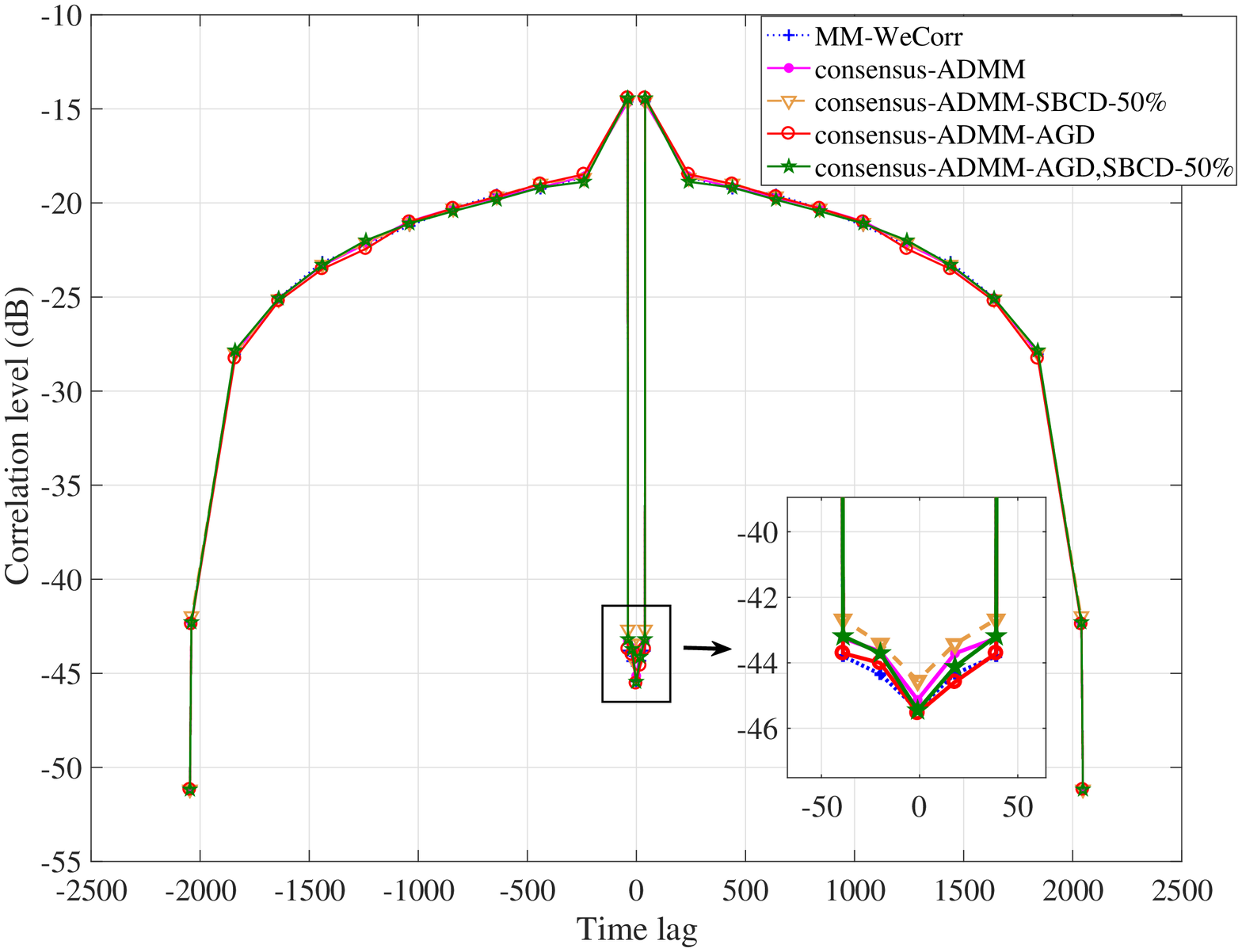}}
\end{minipage}%
}%
\subfigure{
\begin{minipage}[t]{0.5\linewidth}
\centering
\centerline{\includegraphics[scale=0.37]{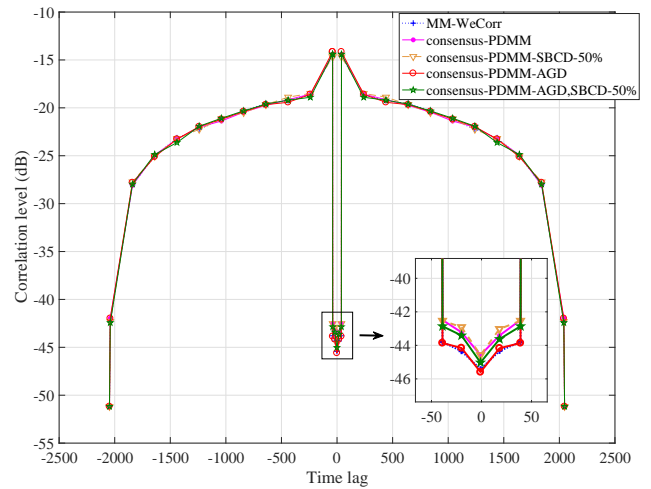}}
\end{minipage}
}
\caption{Correlation levels with $N=2048,M=256,\mathcal{T}=[0,39]$.}
\label{corr 2018 1-39}
\end{figure*}

\begin{figure*}[htbp]
\centering
\subfigure{
\begin{minipage}[t]{0.5\linewidth}
\centering
\centerline{\includegraphics[scale=0.37]{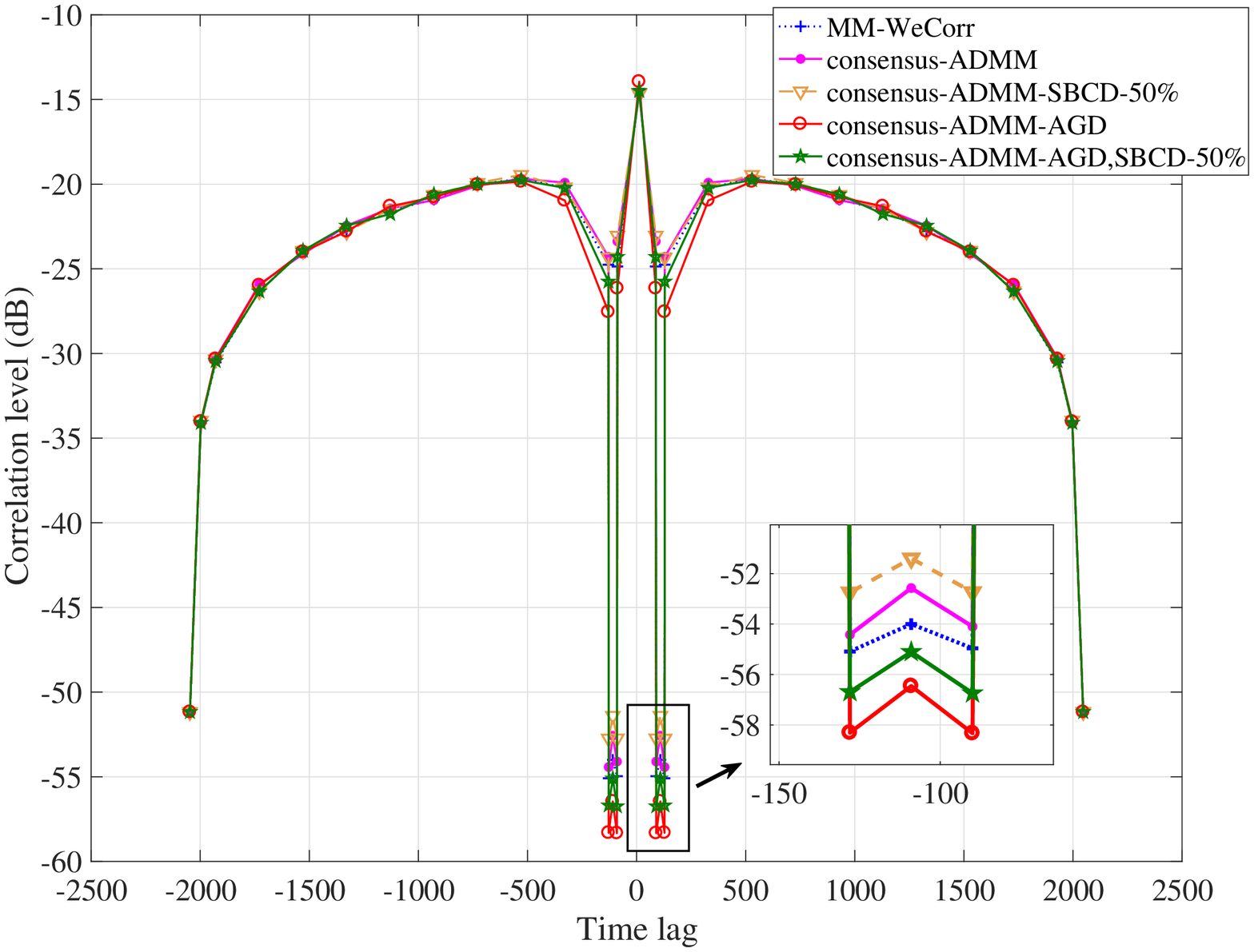}}
\end{minipage}%
}%
\subfigure{
\begin{minipage}[t]{0.5\linewidth}
\centering
\centerline{\includegraphics[scale=0.37]{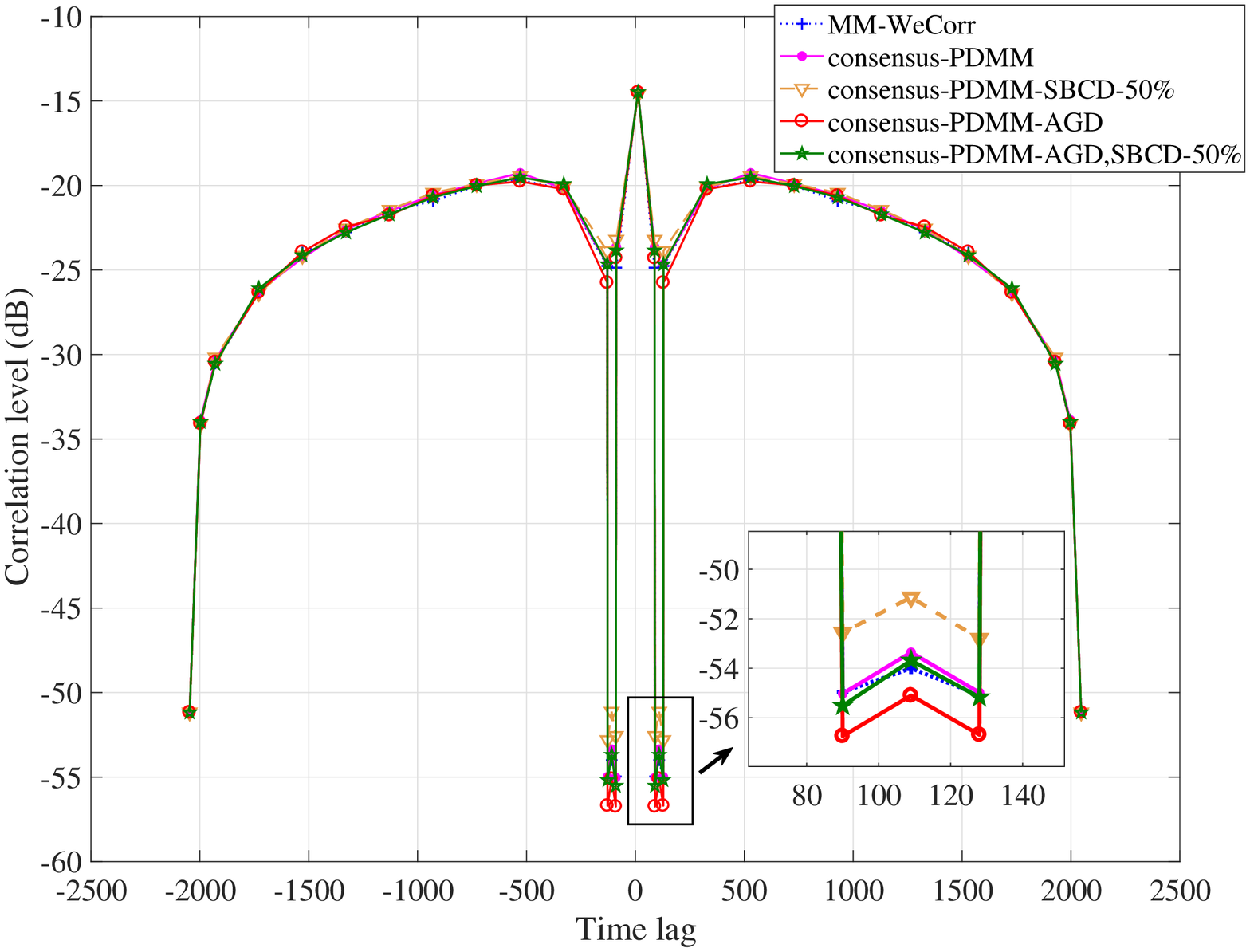}}
\end{minipage}
}
\caption{Correlation levels with $N=2048,M=256,\mathcal{T}=[90,128]$.}
\label{corr 2048 90-128}
\end{figure*}
\begin{table*}[htbp]
\caption{the minimum and average values of the correlation level in $\rm dB$ for interval $[0, 39]$ achieved by different algorithms}
\label{corr level}
\centering
\begin{tabular}{cc|cc|cc|cc|cccccc}
\hline
\hline
 &   & \multicolumn{2}{|c|}{ WeCAN} & \multicolumn{2}{|c|}{ MM-WeCorr} & \multicolumn{2}{|c}{ consensus-ADMM}& \multicolumn{2}{|c}{ consensus-PDMM}      \\
\hline
 $N$ & $M$  & average & minimum & average  & minimum & average & minimum & average & minimum    \\
\hline
 \multirow{2}{*}{256} & 3 & -41.4 & -43.0 & -247.8  & -251.7 & \bf{-279.4} & \bf{-285.8}& \bf{-291.7} & \bf{-295.3}  \\
 & 4 &  -34.3 & -35.1 & -44.2 & -44.7  & \bf{-44.1}  & \bf{-44.5} & \bf{-44.3}  & \bf{-44.7}  \\
\hline
 \multirow{2}{*}{512}  & 4 & -54.5 & -55.5 & -278.6  & -279.0 & \bf{-297.3} & \bf{-298.5}& \bf{-308.3} & \bf{-308.7} \\
 & 8 & -27.6 & -28.6 & -45.8  & -48.1  & \bf{-45.8} & \bf{-47.9}  & \bf{-45.9} & \bf{-48.0} \\
\hline
 \multirow{2}{*}{1024}   & 8& N/A & N/A &-268.5 & -271.1 & \bf{-292.1} & \bf{-293.4} & \bf{-295.7} & \bf{-297.1}   \\
                         & 16& N/A & N/A &-45.7 & -47.3 & \bf{-46.0} & \bf{-48.3} & \bf{-46.0} & \bf{-48.2}\\
\hline
\multirow{2}{*}{2048}   & 16& N/A & N/A &-250.6 & -252.3 & \bf{-299.4} & \bf{-300.1} & \bf{-300.2} & \bf{-301.4}   \\
                         & 32& N/A & N/A &-44.4 & -45.8 & \bf{-44.4} & \bf{-45.8} & \bf{-44.4} & \bf{-46.0}\\
\hline
\hline
\end{tabular}
\end{table*}

Figures \ref{corr 1-39}-\ref{corr 2048 90-128} compare the {\it correlation level} between the proposed consensus-ADMM/PDMM algorithms and the MM-WeCorr approach and WeCan approach. Here, the parameter {\it correlation level} (dB) is defined as
\[
\begin{split}
{\rm correlation ~level}= &20\lg\frac{\|\mathbf{X}^H\mathbf{S}_n\mathbf{X}-N\mathbf{I}\delta_n\|_F^2}{MN^2},~n\in \mathcal{{T}}.
\end{split}
\]
From the figures, we can see that the correlation levels of section $n\geq0$ are symmetrical to that of section $n\leq0$. Compared to WeCAN and MM-WeCorr, the proposed algorithms offer lower correlation levels. This fact is in accordance with the simulation results in Figures \ref{conv fig 256}-\ref{conv fig 2048}.

We tabulate the minimum and average values of the correlation levels achieved by these algorithms in Table \ref{corr level}. For each ($M, N$) pair, the algorithms are repeated $50$ times. The results of different ($M,N$) pairs show that the sequence sets generated by the proposed consensus-ADMM/PDMM algorithms have the optimal correlation property (both minimum and average correlation levels). In addition, we should note that in each iteration, their computational complexities are $\mathcal{O}(M^2N|\mathcal{T}|)$ which are competitive with $\mathcal{O}(M^2N\log N)$ of MM-WeCorr, and smaller than $\mathcal{O}(M^2N^2)$ of WeCAN. However, unlike MM-WeCorr and WeCAN, the consensus-ADMM/PDMM algorithms can be performed in parallel, which means that they are more suitable for large-scale applications from a practical viewpoint of implementation.

\section{Conclusion}
\label{sec:Conclusion}

In this paper, we formulated the unimodular sequences design problem as a consensus-like nonconvex optimization model. Then, two efficient algorithms, named by consensus-ADMM and consensus-PDMM,  were proposed to solve the formulated problem. We proved that, if proper parameters are chosen, the proposed consensus-ADMM algorithm converges and the solution of the consensus-PDMM is guaranteed to be a stationary point of the original problem when it is convergent.
Moreover, we also provided an analysis on the local optimality of the formulated nonconvex optimization problem and computational complexity of the proposed consensus-ADMM/PDMM approaches.
Numerical experiments showed that, compared to the state-of-the-art methods, the proposed algorithms obtained lower correlation sidelobe levels. Besides, the parallel implementation structure let the proposed algorithms be more suitable for large-scale applications.

\appendices
\section{Proof of Lemma 1}
\begin{figure*}[t]
\begin{equation}\label{proof lipschitz}
\begin{split}
 \frac{\|\nabla f_0({\mathbf{\Phi}})-\nabla f_0(\mathbf{\hat{\Phi}})\|^2_F}{\|{\mathbf{\Phi}}-\mathbf{\hat{\Phi}}\|^2_F}
 = \frac{\displaystyle\sum_{i=1}^{N}\sum_{m=1}^M\left|\frac{\partial f_0(\mathbf{\Phi})}{\partial \phi_{i,m}}-\frac{\partial f_0(\mathbf{\hat{\Phi}})}{\partial \hat{\phi}_{i,m}}\right|^2}{\displaystyle\sum_{i=1}^{N}\sum_{m=1}^M|\phi_{i,m}-\hat{\phi}_{i,m}|^2}
 \leq \max_{i,m}\left\{\left|\frac{\frac{\partial f_0(\mathbf{\Phi})}{\partial \phi_{i,m}}-\frac{\partial f_0(\mathbf{\hat{\Phi}})}{\partial \hat{\phi}_{i,m}}}{\phi_{i,m}-\hat{\phi}_{i,m}}\right|^2\right\}.
\end{split}
\end{equation}
\hrulefill
\vspace*{4pt}
\end{figure*}

First, for $\nabla f_0({\mathbf{\Phi}})$, we have the derivations in \eqref{proof lipschitz}. Then, according to the Lagrangian mean value theorem, since $ f_0(\mathbf{\Phi})$ is continuous and differentiable, there exists some point $\bar{\phi}_{i,m}$ between $\phi_{i,m}$ and $\hat\phi_{i,m}$ which satisfies
\begin{equation}\label{mean value theorem}
\frac{\frac{\partial f_0(\mathbf{\Phi})}{\partial \phi_{i,m}}-\frac{\partial f_0(\mathbf{\hat{\Phi}})}{\partial \hat{\phi}_{i,m}}}{\phi_{i,m}-\hat{\phi}_{i,m}}= \frac{\partial^2 f_0(\mathbf{\Phi})}{\partial \bar{\phi}_{i,m}^2}.
\end{equation}
Combining \eqref{proof lipschitz} and \eqref{mean value theorem}, we obtain
\begin{equation}\label{lipschitz phi}
\begin{split}
\!\!\!\!\!\!\frac{\|\nabla f_0(\mathbf{\Phi})\!\!-\!\!\nabla f_0(\mathbf{\hat{\Phi}})\|_F}{\|{\mathbf{\Phi}}-\mathbf{\hat{\Phi}}\|_F}\!\leq \max_{i,m}\!\left\{\!\left|\frac{\partial^2 f_0(\mathbf{\Phi})}{\partial \bar{\phi}_{i,m}^2}\right|\!\right\}\!.
\end{split}
\end{equation}
Moreover, we have \eqref{max second derivative}.
\begin{figure*}
\begin{equation}\label{max second derivative}
\begin{split}
&\left|\frac{\partial^2 f_0(\mathbf{\Phi})}{\partial \bar{\phi}_{i,m}^2}\right|
=\left|2{\rm Re}\left(\!\frac{\partial\mathbf{v}_0^H(\mathbf{\Phi}) }{\partial \bar{\phi}_{i,m}}\frac{\partial\mathbf{v}_0(\mathbf{\Phi}) }{\partial \bar{\phi}_{i,m}}+\!\frac{\partial^2\mathbf{v}_0^H(\mathbf{\Phi}) }{\partial\bar{\phi}_{i,m}^2}(\mathbf{v}_0(\mathbf{\Phi})-\mathbf{c})\!\right)\right|
\leq2\left|\!\frac{\partial\mathbf{v}_0^H(\mathbf{\Phi}) }{\partial \bar{\phi}_{i,m}}\frac{\partial\mathbf{v}_0(\mathbf{\Phi}) }{\partial \bar{\phi}_{i,m}}\right|\!+ \!2\left|\!\frac{\partial^2\mathbf{v}_0^H(\mathbf{\Phi}) }{\partial\bar{\phi}_{i,m}^2}(\mathbf{v}(\mathbf{\Phi})-\mathbf{c})\right|.
\end{split}
\end{equation}
\hrulefill
\vspace*{4pt}
\end{figure*}
From \eqref{XX}, we can see that there are at most $2M-1$ nonzero elements in $\frac{\partial\mathbf{v}_0(\mathbf{\Phi})}{\partial{\bar\phi_{i,m}}}$  and $\frac{\partial^2\mathbf{v}_0^H(\mathbf{\Phi}) }{\partial\bar{\phi}_{i,m}^2}$ respectively. Then, the first term in \eqref{max second derivative} should satisfy the following inequality.
\begin{equation}\label{bound first}
\begin{split}
\left|\!\frac{\partial\mathbf{v}_0^H(\mathbf{\Phi})}{\partial \bar{\phi}_{i,m}}\frac{\partial\mathbf{v}_0(\mathbf{\Phi}) }{\partial \bar{\phi}_{i,m}}\right|\leq 2(M-1).
\end{split}
\end{equation}
Since the maximum modulus of elements in $(\mathbf{v}_0(\mathbf{\Phi})-\mathbf{c})$ is $N$, we can obtain the following inequality for the second term in \eqref{max second derivative}
\begin{equation}\label{bound second}
\left|\!\frac{\partial^2\mathbf{v}_0^H(\mathbf{\Phi}) }{\partial\bar{\phi}_{i,m}^2}(\mathbf{v}(\mathbf{\Phi})-\mathbf{c})\right| \leq 2(M-1)N
\end{equation}
Plugging \eqref{bound first} and \eqref{bound second} into the right side of \eqref{max second derivative}, we have
\begin{equation}\label{}
\left|\frac{\partial^2 f_0(\mathbf{\Phi})}{\partial \bar{\phi}_{i,m}^2}\right|\leq 4(M\!-\!1)(N+1).
\end{equation}

Combining the above results with \eqref{lipschitz phi}, we can see that $\nabla_{\mathbf{\Phi}}f_0(\mathbf{\Phi})$ is Lipschitz continuous with the constant $L_0>4(M-1)(N+1)$.

Second, for $\nabla f_n(\mathbf{\Phi})$, there exists
\begin{equation}\label{lipschitz phi n}
\begin{split}
\frac{\|\nabla f_n({\mathbf{\Phi}})-\nabla f_n(\mathbf{\hat{\Phi}})\|_F}{\|{\mathbf{\Phi}}\!-\!\mathbf{\hat{\Phi}}\|_F} \!\leq \max_{i,m}\left\{\left|\frac{\partial^2 f_n(\mathbf{\Phi})}{\partial\bar{\phi}_{i,m}^2} \right|\right\}.
\end{split}
\end{equation}
For $\frac{\partial^2 f_n(\mathbf{\Phi})}{\partial\bar{\phi}_{i,m}^2} $, we have
\begin{equation}\label{second derivative phi n}
\begin{split}
\hspace{-0.45cm}&\left|\!\frac{\partial^2 f_n(\mathbf{\Phi})}{\partial\bar{\phi}_{i,m}^2} \!\right|\!\!=\!\!\left|2{\rm Re}\!\left(\!
\frac{\partial\mathbf{v}_n^H(\mathbf{\Phi}) }{\partial \bar{\phi}_{i,m}}\!\frac{\partial\mathbf{v}_n(\mathbf{\Phi}) }{\partial \bar{\phi}_{i,m}}\!+\!
\mathbf{v}_n^H(\mathbf{\Phi})\frac{\partial^2\mathbf{v}_n(\mathbf{\Phi}) }{\partial \bar{\phi}_{i,m}^2} \!\right)\!\right|\\
\hspace{-0.45cm}&\leq \!2\left|
\frac{\partial\mathbf{v}_n^H(\mathbf{\Phi}) }{\partial \bar{\phi}_{i,m}}\!\frac{\partial\mathbf{v}_n(\mathbf{\Phi})}{\partial \bar{\phi}_{i,m}}\right|+2\left|\mathbf{v}_n^H(\mathbf{\Phi})\frac{\partial^2\mathbf{v}_n(\mathbf{\Phi}) }{\partial \bar{\phi}_{i,m}^2}\right|.
\end{split}
\end{equation}
Through similar derivations to \eqref{lipschitz phi}, we have
\begin{equation}\label{bound phi n}
\left|\frac{\partial^2 f_n(\mathbf{\Phi})}{\partial\bar{\phi}_{i,m}^2} \right| \leq 4(M-1)(N+1),
\end{equation}
which results in gradients $\nabla f_n(\mathbf{\Phi}), n\in\mathcal{{T}}$ being Lipschitz continuous with constant $L_n > 4(M-1)(N+1)$.  $\hfill\blacksquare$

\section{Proofs of Several Lemmas for the Proposed Consensus-ADMM algorithm}
\begin{lemma}\label{upperbound function}
For the upper-bounded function $\mathcal{U}_n\!\left(\!\mathbf{\Phi}^{k}, \mathbf{\Phi}_n,\mathbf{\Lambda}_n^k\right)$ defined in \eqref{upperbound Un}, we have the following inequality
\begin{equation}\label{difference Un}
\begin{split}
&\mathcal{U}_n\left(\mathbf{\Phi}^{k+1}, \mathbf{\Phi}_n,\mathbf{\Lambda}_n^k\right)-\mathcal{L}_n(\mathbf{\Phi}^{k+1}, \mathbf{\Phi}_n, \mathbf{\Lambda}_n^k)\\
\leq&\ 2L_n\|{\mathbf{\Phi}_n}-{\mathbf{\Phi}^{k+1}}\|_F^2, \ \ \forall n\in \mathcal{{T}}.
\end{split}
\end{equation}
\end{lemma}
\begin{proof} Based on \eqref{upperbound Un}, we have
\begin{equation}\label{difference phi}
\begin{split}
\!\!\!\! &\!\mathcal{U}_n\left(\mathbf{\Phi}^{k+1}, \mathbf{\Phi}_n,\mathbf{\Lambda}_n^k\right)-\mathcal{L}_n(\mathbf{\Phi}^{k+1}, \mathbf{\Phi}_n, \mathbf{\Lambda}_n^k)\\
\!\!=& f_n(\mathbf{\Phi}^{k+1})\!-\!f_n(\mathbf{\Phi}_n)\!+\!\langle \nabla f_n( {\mathbf{\Phi}}^{k+1}),\mathbf\Phi_n\!-\!\mathbf\Phi^{k+1}\rangle \\
\!\!+&\frac{L_n}{2}\|{\mathbf{\Phi}}_n\!-\!\!{\mathbf{\Phi}}^{k+1}\|_F^2.
\end{split}
\end{equation}
Since $\nabla f_n(\mathbf{\Phi})$ is Lipschitz continuous, there exists
\begin{equation}\label{lipschitz property}
\begin{split}
&f_n(\mathbf{\Phi}^{k+1})-f_n(\mathbf{\Phi}_n)   \\
\leq &\ \langle \nabla f_n({\mathbf{\Phi}}_n),{\mathbf{\Phi}^{k+1}}-{\mathbf{\Phi}_n}\rangle+\frac{L_n}{2}\|{\mathbf{\Phi}}^{k+1}-{\mathbf{\Phi}}_n\|_F^2.
\end{split}
\end{equation}
Plugging the above inequality into \eqref{difference phi}, we can get
\begin{equation}\label{diff phi}
\begin{split}
\hspace{-0.3cm}&\ \ \mathcal{U}_n\left(\mathbf{\Phi}^{k+1}, \mathbf{\Phi}_n,\mathbf{\Lambda}_n^k\right)-\mathcal{L}_n(\mathbf{\Phi}^{k+1}, \mathbf{\Phi}_n, \mathbf{\Lambda}_n^k)\\
\hspace{-0.3cm}\leq& \langle \nabla f_n(\mathbf{\Phi}^{k+1})-\nabla f_n(\mathbf{\Phi}_n),{\mathbf{\Phi}}^{k+1}-{\mathbf{\Phi}}_n\rangle \\
&+L_n\|{\mathbf{\Phi}^{k+1}-{\mathbf{\Phi}}_n}\|_F^2.
\end{split}
\end{equation}
Furthermore, according to Lemma \ref{Lipschtiz continuous}, there exists
\[
 \begin{split}
   &\langle \nabla f_n(\mathbf{\Phi}^{k+1})\!- \!\nabla f_n( \mathbf{\Phi}_n),{\mathbf{\Phi}}^{k+1}\!-\!{\mathbf{\Phi}_n}\rangle
   \leq L_n\|\mathbf{\Phi}^{k+1}\!-\!\mathbf{\Phi}_n\|^2_F.\\
 \end{split}
\]
Plugging it into \eqref{diff phi}, we can get \eqref{difference Un}. This completes the proof.
$\hfill\blacksquare$
\end{proof}

\begin{lemma}\label{lemma difference lambda}
In each consensus-ADMM iteration, $\forall n\in \mathcal{{T}} $, $\|{\bf\Lambda}_n^{k+1}-{\bf\Lambda}_n^k\|_F^2$ is upper-bounded as
\begin{equation}\label{diff lambda ADMM}
\begin{split}
\hspace{-0.3cm}\|{\mathbf{\Lambda}_n^{k+1}}\!\!-\!\!{\mathbf{\Lambda}_n^k}\|_F^2 \!\!\leq\! 2L_n^2\!\left(2\|{\mathbf{\Phi}_n^{k+1}}\!\!\!-\!{\mathbf{\Phi}_n^k}\|_F^2\!\!+\!3\|{\mathbf{\Phi}^{k+1}}\!\!-\!\!{\mathbf{\Phi}^{k}}\|_F^2\!\right).
\end{split}
\end{equation}
\end{lemma}

\begin{proof}
The optimal solutions of problems \eqref{step2 ADMM} can be obtained by solving the linear equations $\nabla_{\mathbf{\Phi}_n}\mathcal{U}_n\left(\mathbf{\Phi}^{k+1},\mathbf{\Phi}_n,\mathbf{\Lambda}_n^k\right)=0$, $\forall\ n\in \mathcal{{T}}$, i.e.,
\begin{equation}\label{solve admm b}
\begin{split}
\nabla f_n({\mathbf{\Phi}}^{k+1})+ \mathbf{\Lambda}_n^k+(\rho_n+L_n)(\mathbf{\Phi}_n^{k+1}-\mathbf{\Phi}^{k+1})=0.
\end{split}
\end{equation}
Combining it with \eqref{step3 ADMM}, we can get
\begin{equation}\label{lambda update}
\begin{split}
\!\!\!\!\!\mathbf{\Lambda}_n^{k+1}=-\nabla f_n({\mathbf{\Phi}}^{k+1})- L_n(\mathbf{\Phi}_n^{k+1}-\mathbf{\Phi}^{k+1}).
\end{split}
\end{equation}
Plugging \eqref{lambda update} into  $\|\mathbf{\Lambda}_n^{k+1}\!-\!\mathbf{\Lambda}_n^{k}\|_F^2$, we have the following derivations
\[
\begin{split}
\!\!\!&\|\mathbf{\Lambda}_n^{k+1}-\mathbf{\Lambda}_n^{k}\|_F^2\\
\!\!\!=&\|\nabla f_n(\!{\mathbf{\Phi}^{k+1}}\!)\!-\!\nabla f_n({\mathbf{\Phi}^{k}})\!+\!L_n(\mathbf{\Phi}_n^{k+1} \!-\!\mathbf{\Phi}^{k+1}\!-\!\mathbf{\Phi}_n^{k}\!+\!\mathbf{\Phi}^{k})\|_F^2\\
\!\!\!{\leq}&2\|\!\nabla\! f_n(\!{\mathbf{\Phi}^{k+1}}\!)\!-\!\!\nabla \!f_n(\!{\mathbf{\Phi}^{k}})\!\|_F^2\!+\!\!2L_n^2\|\!\mathbf{\Phi}_n^{k+1} \!-\!\mathbf{\Phi}_n^{k}\!-\!\mathbf{\Phi}^{k+1}\!+\!\mathbf{\Phi}^{k}\!\|_F^2\\
\!\!\!{\leq}&\ 2L_n^2(2\|{\mathbf{\Phi}_n^{k+1}}-{\mathbf{\Phi}_n^{k}}\|_F^2+3\|{\mathbf{\Phi}^{k+1}}-{\mathbf{\Phi}^{k}}\|_F^2),
\end{split}
\]
where the second inequality comes from the Lipschitz continuity of function $\nabla\! f_n(\mathbf{\Phi})$. This completes the proof. $\hfill\blacksquare$
\end{proof}

\begin{lemma}\label{successive difference}
In each consensus-ADMM iteration, if
\begin{equation}\label{cn}
\begin{split}
&\bar{c}_{n}\!=\!\rho_n^3\!-\!7\rho_n^2L_n\!-\!8\rho_nL_n^2\!-\!32L_n^3\geq0 ,\\
& \tilde{c}_{n}\!=\! \rho_n^3\!-\!12\rho_nL_n^2\!-\!48L_n^3\geq0.
\end{split}
\end{equation}
then, $\mathcal{L}(\mathbf{\Phi}^k\!, \!\{\mathbf{\Phi}_n^k,\mathbf{\Lambda}_n^k, \! n\in\mathcal{{T}}\})$ {\it decreases sufficiently}, i.e.,
\begin{equation}\label{succ-diff phi}
\begin{split}
\hspace{-0.4cm}& \ \ \mathcal{L}(\mathbf{\Phi}^k\!, \!\{\mathbf{\Phi}_n^k,\mathbf{\Lambda}_n^k, \! n\in\mathcal{{T}}\})\!-\!\mathcal{L}(\mathbf{\Phi}^{k\!+\!1},\!\{\mathbf{\Phi}_n^{k\!+\!1}\!,\!\mathbf{\Lambda}_n^{k\!+\!1}\!,\! n\in\mathcal{{T}}\})\\
\hspace{-0.4cm}&\geq \!\sum_{n\in \mathcal{{T}}}\frac{1}{2\rho_n^2}\left( \bar{c}_{n}\|{\mathbf{\Phi}_n^{k+1}}\!-\!{\mathbf{\Phi}_n^k}\|_F^2\!+\!\tilde{c}_{n} \|{\mathbf{\Phi}^{k+1}}\!-\!{\mathbf{\Phi}^{k}}\|_F^2 \right),
\end{split}
\end{equation}
\end{lemma}
\begin{proof}
To facilitate the subsequent derivations, we define the following quantities
\[
  \begin{split}
    &\Delta_{\mathbf{\Phi}}^k\!=\!\mathcal{L}(\mathbf{\Phi}^{k}\!, \{\mathbf{\Phi}_n^k,\mathbf{\Lambda}_n^k, n\in\mathcal{{T}}\})-\mathcal{L}(\mathbf{\Phi}^{k+1}\!,\!\{\mathbf{\Phi}_n^{k},\!\mathbf{\Lambda}_n^{k},n\in\mathcal{{T}}\}), \\
    &\Delta_{\mathbf{\Phi}_n}^k\!\!=\!\mathcal{L}_n(\mathbf{\Phi}^{k+1}\!, \mathbf{\Phi}_n^k,\mathbf{\Lambda}_n^k)-\!\mathcal{L}_n(\mathbf{\Phi}^{k+1}\!,\!\mathbf{\Phi}_n^{k+1}\!,\!\mathbf{\Lambda}_n^{k}),\\
    &\Delta_{\mathbf{\Lambda}_n}^k \!\!=\!\mathcal{L}_n(\mathbf{\Phi}^{k+1}\!, \!\mathbf{\Phi}_n^{k+1}\!,\!\mathbf{\Lambda}_n^k)\!-\!\mathcal{L}_n(\mathbf{\Phi}^{k+1}\!,\mathbf{\Phi}_n^{k+1}\!,\!\mathbf{\Lambda}_n^{k+1}).
  \end{split}
\]
Then, from the above quantities, we have
\begin{equation}\label{diff L}
  \begin{split}
\hspace{-0.2cm}& \mathcal{L}(\mathbf{\Phi}^k\!, \!\{\mathbf{\Phi}_n^k,\!\mathbf{\Lambda}_n^k, \! n\!\in\!\mathcal{{T}}\})\!\!-\!\!\mathcal{L}(\mathbf{\Phi}^{k+1},\!\{\mathbf{\Phi}_n^{k+1}\!,\!\mathbf{\Lambda}_n^{k+1}\!,\! n\!\in\!\mathcal{{T}}\})\\
\hspace{-0.2cm}=& \Delta_{\mathbf{\Phi}}^k + \displaystyle\sum_{n\in\mathcal{{T}}} \left( \Delta_{\mathbf{\Phi}_n}^k+\Delta_{\mathbf{\Lambda}_n}^k\right).
  \end{split}
\end{equation}
Since $\mathcal{L}\left(\mathbf{\Phi},\{\mathbf{\Phi}_n^{k},\mathbf{\Lambda}_n^{k}, n\in\mathcal{{T}}\}\right)$ with respect to $\mathbf{\Phi}$ is strongly convex, $\Delta_{\mathbf{\Phi}}^k$ should satisfy the following inequality
\begin{equation}\label{delta phi 1}
\begin{split}
\Delta_{\mathbf{\Phi}}^k&\geq\left\langle\nabla_{\mathbf{\Phi}}\mathcal{L}(\mathbf{\Phi}^{k+1}\!,\!\{\mathbf{\Phi}_n^{k},\!\mathbf{\Lambda}_n^{k},n\in\mathcal{{T}}\}),\mathbf{\Phi}^k-\mathbf{\Phi}^{k+1}\right\rangle\\
&\ \ +\sum\limits_{n\in \mathcal{{T}}}\frac{\rho_n}{2}\|\mathbf{\Phi}^{k+1}\!-\!\mathbf{\Phi}^{k}\|_F^2.
\end{split}
\end{equation}
Moreover, since
$\mathbf{\Phi}^{k\!+\!1}\!\!=\!\!\underset{0\preceq\mathbf{\Phi}\prec2\pi} {\arg \min} \mathcal{L}\!\left(\!\mathbf{\Phi},  \{\mathbf{\Phi}_n^k,\mathbf{\Lambda}_n^k,n\!\in\!\mathcal{T}\}\!\right)$, there exists
\[
\left\langle\nabla_{\mathbf{\Phi}}\mathcal{L}(\mathbf{\Phi}^{k+1}\!,\!\{\mathbf{\Phi}_n^{k},\!\mathbf{\Lambda}_n^{k},n\in\mathcal{{T}}\}),\mathbf{\Phi}^k-\mathbf{\Phi}^{k+1}\right\rangle\geq0.
\]
Plugging it into \eqref{delta phi 1}, we can obtain
\begin{equation}\label{delta alpha phi 1}
\begin{split}
\Delta_{\mathbf{\Phi}}^k&\geq\sum\limits_{n\in \mathcal{{T}}}\frac{\rho_n}{2}\|\mathbf{\Phi}^{k+1}\!-\!\mathbf{\Phi}^{k}\|_F^2.
\end{split}
\end{equation}

Similarly, since $\mathcal{L}_n(\mathbf{\Phi}^{k+1}, \mathbf{\Phi}_n,\mathbf{\Lambda}_n^k)$ is strongly convex with respect to $\mathbf{\Phi}_n$, $\Delta_{\mathbf{\Phi}_n}^k$ should satisfy
\begin{equation}\label{delta phi n 1}
\begin{split}
\Delta_{\mathbf{\Phi}_n}^k&\geq\mathcal{L}_n(\mathbf{\Phi}^{k+1}\!, \mathbf{\Phi}_n^k,\mathbf{\Lambda}_n^k)-\mathcal{U}_n(\mathbf{\Phi}^{k+1}\!,\mathbf{\Phi}_n^{k+1},\mathbf{\Lambda}_n^{k}).
\end{split}
\end{equation}
Moreover, according to Lemma \ref{upperbound function} and the strong convexity of the functions $\mathcal{U}_n\left(\mathbf{\Phi}^{k+1},\mathbf{\Phi}_n,\mathbf{\Lambda}_n^{k}\right)$, $n\in\mathcal{{T}}$ with respect to $\mathbf{\Phi}_n$, we have the following two inequalities respectively
\[
  \begin{split}
  &\mathcal{L}_n(\mathbf{\Phi}^{k+1}\!, \!\mathbf{\Phi}_n^k,\!\mathbf{\Lambda}_n^k)\!-\!\mathcal{U}_n(\mathbf{\Phi}^{k+1}\!,\mathbf{\Phi}_n^{k},\mathbf{\Lambda}_n^{k})
  \!\!\geq\!\! -2L_n\!\|{\mathbf{\Phi}_n^k}\!-\!{\mathbf{\Phi}^{k+1}}\|_F^2, \\
  &\mathcal{U}_n(\mathbf{\Phi}^{k+1}\!\!, \!\mathbf{\Phi}_n^k,\!\mathbf{\Lambda}_n^k)\!\!-\!\!\mathcal{U}_n(\mathbf{\Phi}^{k+1}\!\!,\!\mathbf{\Phi}_n^{k+1}\!\!,\!\mathbf{\Lambda}_n^{k}) \!\!\geq\!\!\frac{\rho_n\!+\!L_n}{2}\|\mathbf{\Phi}_n^{k+1}\!\!\!-\!\!\mathbf{\Phi}_n^{k}\|_F^2.
  \end{split}
\]
Plugging them into \eqref{delta phi n 1}, it can be changed to
\begin{equation}
\begin{split}
\!\!\!\!\Delta_{\mathbf{\Phi}_n}^k&\!\geq\!-2L_n\!\|{\mathbf{\Phi}_n^k}\!-\!{\mathbf{\Phi}^{k+1}}\|_F^2\!+\!\frac{\rho_n\!+\!L_n}{2}\|\mathbf{\Phi}_n^{k+1}\!-\!\mathbf{\Phi}_n^{k}\|_F^2,\\
\vspace{-20pt}&\!\geq\!-4L_n\!\|{\mathbf{\Phi}_n^{k+1}}\!\!-\!{\mathbf{\Phi}^{k+1}}\|_F^2\!+\!\frac{\rho_n\!\!-\!7\!L_n}{2}\|\mathbf{\Phi}_n^{k+1}\!\!-\!\!\mathbf{\Phi}_n^{k}\|_F^2.\\
\end{split}
\end{equation}
Since $\mathbf{\Lambda}_n^{k+1} \!-\! \mathbf{\Lambda}_n^{k} = \rho_n(\mathbf{\Phi}_n^{k+1}\! -\mathbf{\Phi}^{k+1})$, the above inequality can be rewritten as
\begin{equation}\label{delta phi n}
\!\Delta_{\mathbf{\Phi}_n}^k\!\geq\!\frac{-4L_n}{\rho_n^2}\!\|\mathbf{\Lambda}_n^{k+1}-\mathbf{\Lambda}_n^{k}\|_F^2\!+\!\frac{\rho_n\!-7\!L_n}{2}\|\mathbf{\Phi}_n^{k+1}\!-\!\mathbf{\Phi}_n^{k}\|_F^2.
\end{equation}
Furthermore, plugging \eqref{diff lambda ADMM} into \eqref{delta phi n}, it can be derived as
\begin{equation}\label{bound delta phi n}
\!\Delta_{\mathbf{\Phi}_n}^k\!\!\!\geq\!\!\frac{\rho_n^3\!\!-\!\!7\rho_n^2L_n\!\!-\!\!32L_n^3}{2\rho_n^2}\!\|\mathbf{\Phi}_n^{k+1}\!\!\!-\!\mathbf{\Phi}_n^{k}\|_F^2\!
-\!\frac{24L_n^3}{\rho_n^2}\|\mathbf{\Phi}^{k+1}\!\!\!-\!\mathbf{\Phi}^{k}\|_F^2.
\end{equation}

For $\Delta_{\mathbf{\Lambda}_n}^k$, through similar derivations and the results in Lemma \ref{lemma difference lambda}, there exists
\begin{equation}\label{delta lambda}
\Delta_{\mathbf{\Lambda}_n}^k {\geq}-\frac{2L_n^2}{\rho_n}\left(2\|{\mathbf{\Phi}_n^{k+1}}-{\mathbf{\Phi}_n^{k}}\|_F^2+3\|{\mathbf{\Phi}^{k+1}}-{\mathbf{\Phi}^{k}}\|_F^2\right).
\end{equation}

Plugging \eqref{delta alpha phi 1}, \eqref{bound delta phi n}, and \eqref{delta lambda} into \eqref{diff L}, we have the following inequality
\[
  \begin{split}
\hspace{-0.4cm}& \ \ \mathcal{L}(\mathbf{\Phi}^k\!, \!\{\mathbf{\Phi}_n^k,\mathbf{\Lambda}_n^k, \! n\in\mathcal{{T}}\})\!-\!\mathcal{L}(\mathbf{\Phi}^{k\!+\!1},\!\{\mathbf{\Phi}_n^{k\!+\!1}\!,\!\mathbf{\Lambda}_n^{k\!+\!1}\!,\! n\in\mathcal{{T}}\})\\
\hspace{-0.4cm}&\geq \sum_{n\in \mathcal{{T}}}\frac{1}{2\rho_n^2}\left( \bar{c}_{n}\|{\mathbf{\Phi}_n^{k+1}}\!-\!{\mathbf{\Phi}_n^k}\|_F^2\!+\!\tilde{c}_{n} \|{\mathbf{\Phi}^{k+1}}\!-\!{\mathbf{\Phi}^{k}}\|_F^2 \right),
  \end{split}
\]
where $\bar{c}_{n}$ and $\tilde{c}_{n}$ are defined in \eqref{cn}. This completes the proof. $\hfill\blacksquare$
\end{proof}

\begin{lemma}\label{lemma lower bound}
 If $\rho_n>5L_n$, augmented Lagrangian function
 \begin{equation}\label{Lgeq0}
 \mathcal{L}(\mathbf{\Phi}^{k+1},\{\mathbf{\Phi}_n^{k+1},\mathbf{\Lambda}_n^{k+1}, n\in\mathcal{{T}}\})\geq 0 , \forall k.
 \end{equation}
 \end{lemma}
\begin{proof}
 First, plugging \eqref{lambda update} into $\mathcal{L}_n({\mathbf{\Phi}}^{k+1},\mathbf{\Phi}_n^{k+1},\mathbf{\Lambda}_n^{k+1})$, it can be written as
\begin{equation}\label{lag bound}
\begin{split}
&\mathcal{L}_n\left(\mathbf{\Phi}^{k+1}, \mathbf{\Phi}_n^{k+1},\mathbf{\Lambda}_n^{k+1}\right) \\
=&  f_n({\mathbf{\Phi}_n^{k+1}})+(\frac{\rho_n}{2}-L_n)\|\mathbf{\Phi}_n^{k+1}-\mathbf{\Phi}^{k+1}\|_F^2 \\
&\hspace{2.8cm}+\langle\nabla  f_n({\mathbf{\Phi}}^{k+1}), \mathbf{\Phi}^{k+1}-\mathbf{\Phi}_n^{k+1}\rangle.
\end{split}
\end{equation}
Since $\|\nabla f_n({\mathbf{\Phi}}^{k+1})\!-\!\nabla  f_n({\mathbf{\Phi}}_n^{k+1})\|_F\!\leq\!L_n\|{\mathbf{\Phi}}_n^{k+1}\!-\!{\mathbf{\Phi}}^{k+1}\|_F$, we further have the following inequality
\[
  \begin{split}
    &\langle\nabla  f_n({\mathbf{\Phi}}^{k+1}), \mathbf{\Phi}^{k+1}-\mathbf{\Phi}_n^{k+1}\rangle \\
    \geq& \langle\nabla f_n({\mathbf{\Phi}}_n^{k+\!1}), \mathbf{\Phi}^{k+1}-\mathbf{\Phi}_n^{k+1}\rangle -L_n\|\mathbf{\Phi}_n^{k+1}-\mathbf{\Phi}^{k+1}\|^2_F.
  \end{split}
\]
Replacing the last term in \eqref{lag bound} with the above inequality, we can obtain
\begin{equation}\label{lag bound 2}
\begin{split}
& \mathcal{L}_n\left(\mathbf{\Phi}^{k+1}, \mathbf{\Phi}_n^{k+1},\mathbf{\Lambda}_n^{k+1}\right) \\
\geq& f_n({\mathbf{\Phi}_n^{k+1}})+\langle\nabla f_n({\mathbf{\Phi}}_n^{k+\!1}), \mathbf{\Phi}^{k+1}-\mathbf{\Phi}_n^{k+1}\rangle \\
&\hspace{2.5cm} +(\frac{\rho_n}{2}-2L_n)\|\mathbf{\Phi}_n^{k+1}-\mathbf{\Phi}^{k+1}\|_F^2.
\end{split}
\end{equation}
Since $\nabla f_n(\mathbf{\Phi}_n)$ is Lipschitz continuous, there exists
\[
  \begin{split}
  &f_n({\mathbf{\Phi}^{k+1}})\leq f_n({\mathbf{\Phi}_n^{k+1}})+\langle\nabla f_n({\mathbf{\Phi}}_n^{k+\!1}), \mathbf{\Phi}^{k+1}-\mathbf{\Phi}_n^{k+1}\rangle \\
  &\hspace{4.4cm}+ \frac{L_n}{2}\|\mathbf{\Phi}_n^{k+1}-\mathbf{\Phi}^{k+1}\|_F^2.
  \end{split}
\]
Replacing the first two terms in right hand side of \eqref{lag bound 2} through the above inequality, it can be simplified as
\begin{equation}\label{lag bound 3}
\begin{split}
&\mathcal{L}_n\left(\mathbf{\Phi}^{k+1}, \mathbf{\Phi}_n^{k+1},\mathbf{\Lambda}_n^{k+1}\right)\\
\geq& f_n({\mathbf{\Phi}^{k+1}})+ \frac{\rho_n-5L_n}{2}\|\mathbf{\Phi}_n^{k+1}-\mathbf{\Phi}^{k+1}\|_F^2.
\end{split}
\end{equation}
Second, since
\[
\begin{split}
\mathcal{L}(\mathbf{\Phi},\!\{\mathbf{\Phi}_n,\!\mathbf{\Lambda}_n,\!n\in \mathcal{{T}}\})
\!= \!\sum\limits_{n\in \mathcal{{T}} }\! \left(\mathcal{L}_n\left(\mathbf{\Phi}^{k+1}, \mathbf{\Phi}_n^{k+1},\mathbf{\Lambda}_n^{k+1}\right)\right),
\end{split}
\]
we can get
\begin{equation}\label{upp func bound}
\begin{split}
&\mathcal{L}(\mathbf{\Phi}^{k+1},\{\mathbf{\Phi}_n^{k+1},\mathbf{\Lambda}_n^{k+1}, n\in\mathcal{{T}}\})\\
 \geq& \sum_{n\in \mathcal{{T}}}\left(\!f_n(\mathbf{\Phi}^{k+1})+\frac{\rho_n\!-\!5L_n}{2}\|\mathbf{\Phi}_n^{k+1}\!-\!\mathbf{\Phi}^{k+1}\|_F^2\right).
\end{split}
\end{equation}
Since $\forall n\in\mathcal{{T}}, f_n(\mathbf{\Phi})\geq0$, we can conclude that, if $\rho_n>5L_n$, $\forall k$, $\mathcal{L}(\mathbf{\Phi}^{k+1},\{\mathbf{\Phi}_n^{k+1},\mathbf{\Lambda}_n^{k+1}, n\in\mathcal{{T}}\})> 0$. This completes the proof. $\hfill\blacksquare$
\end{proof}


\section{Proof of Theorem \ref{theorem admm}}\label{proof theorem 1}

First, we prove \eqref{conv variables} in Theorem \ref{theorem admm}.

In Lemmas \ref{successive difference}-\ref{lemma lower bound}, we desire
$\rho_n^3\!-\!7\rho_n^2L_n\!-\!8\rho_nL_n^2\!-\!32L_n^3>0$, $    \rho_n^3\!-\!12\rho_nL_n^2\!-\!48L_n^3>0$, and $\rho_n\geq5L_n$ hold,
where the first two inequalities can guarantee augmented Lagrangian function $\mathcal{L}(\cdot^k)$ {\it decreases sufficiently} and the last one
 can guarantee $\mathcal{L}(\cdot^k)\geq0$ in every iteration.
Through the famous Cardano formula \cite{Cardano}, we can obtain that the first two inequalities hold
when $\rho_n\geq8.41L_n$ and $\rho_n\geq4.72L_n$. Combining them with $\rho_n\geq5L_n$, we can see that when $\rho_n\geq8.41L_n$, all the inequalities hold simultaneously. To simplify the description, we choose  $\forall n \in \mathcal{{T}} ,\rho_n\geq 9L_n$, which can guarantee that
  \eqref{succ-diff phi} and \eqref{Lgeq0} in Lemma \ref{successive difference} and Lemma \ref{lemma lower bound} hold simultaneously.

 Summing both sides of the inequality \eqref{succ-diff phi} at $k=1, 2,\dotsb, +\infty$, we can obtain
 \[
\begin{split}
& \mathcal{L}(\mathbf{\Phi}^1\!, \!\{\mathbf{\Phi}_n^1,\mathbf{\Lambda}_n^1, \! n\in\mathcal{{T}}\})\!-\!\!\!\underset{k\rightarrow +\infty}\lim\mathcal{L}(\mathbf{\Phi}^{k\!+\!1},\!\{\mathbf{\Phi}_n^{k\!+\!1}\!,\!\mathbf{\Lambda}_n^{k\!+\!1}\!,\! n\in\mathcal{{T}}\})\\
&\geq \!\sum_{k=1}^{+\infty}\sum_{n\in \mathcal{{T}}}\frac{1}{2\rho_n^2}\left( \bar{c}_{n}\|{\mathbf{\Phi}_n^{k+1}}\!-\!{\mathbf{\Phi}_n^k}\|_F^2\!+\!\tilde{c}_{n} \|{\mathbf{\Phi}^{k+1}}\!-\!{\mathbf{\Phi}^{k}}\|_F^2 \right).
\end{split}
\]
 Since \eqref{Lgeq0} holds, the following inequality holds.
 \[
\begin{split}
\hspace{-0.4cm}& \ \ \mathcal{L}(\mathbf{\Phi}^1\!, \!\{\mathbf{\Phi}_n^1,\mathbf{\Lambda}_n^1, \! n\in\mathcal{{T}}\})\\
\hspace{-0.4cm}&\geq \!\sum_{k=1}^{+\infty}\sum_{n\in \mathcal{{T}}}\frac{1}{2\rho_n^2}\left( \bar{c}_{n}\|{\mathbf{\Phi}_n^{k+1}}\!-\!{\mathbf{\Phi}_n^k}\|_F^2\!+\!\tilde{c}_{n} \|{\mathbf{\Phi}^{k+1}}\!-\!{\mathbf{\Phi}^{k}}\|_F^2 \right).
\end{split}
\]

Since $\bar{c}_{n}, \tilde{c}_{n}>0$ and $\mathcal{L}(\mathbf{\Phi}^1\!, \!\{\mathbf{\Phi}_n^1,\mathbf{\Lambda}_n^1, \! n\in\mathcal{{T}}\})$ is finite, we can conclude that \eqref{convergence limit PHI} and \eqref{convergence limit PHI n} hold.
\begin{equation} \lim\limits_{k\rightarrow+\infty}\|\mathbf{\Phi}^{k+1}-\mathbf{\Phi}^{k}\|_F= 0. \label{convergence limit PHI}
\end{equation}
\begin{equation}
\lim\limits_{k\rightarrow+\infty}\|\mathbf{\Phi}_n^{k+1}-\mathbf{\Phi}_n^{k}\|_F= 0, ~\forall~ n\in \mathcal{{T}}. \label{convergence limit PHI n}
\end{equation}
Plugging \eqref{convergence limit PHI} and \eqref{convergence limit PHI n} into \eqref{diff lambda ADMM}, there exists
\begin{equation}\label{convergence Lambda}
\lim\limits_{k\rightarrow+\infty}\|\mathbf{\Lambda}_n^{k+1}\!-\!\mathbf{\Lambda}_n^{k}\|_F\! = \!0, ~\forall~ n\in \mathcal{{T}}.
\end{equation}
Plugging \eqref{convergence Lambda} into \eqref{step3 ADMM}, we further have
\begin{equation}\label{convergence diff}
\lim\limits_{k\rightarrow+\infty}\|\mathbf{\Phi}_n^{k+1}-\mathbf{\Phi}^{k+1}\|_F = 0, ~\forall~ n\in \mathcal{{T}}.
\end{equation}
Since $0\preceq\mathbf{\Phi}\prec2\pi$, \eqref{convergence limit PHI} indicates $\mathbf{\Phi}^k$ converges to some limit point as $k\rightarrow +\infty$, i.e.,
\begin{equation}\label{convergence_Phi}
  \lim\limits_{k\rightarrow+\infty}\mathbf{\Phi}^{k}=\mathbf{\Phi}^{*}.
\end{equation}
Combining the above result with \eqref{convergence limit PHI n} and \eqref{convergence diff}, we can obtain
\begin{equation}\label{convergence_Phi_n}
   \lim\limits_{k\rightarrow+\infty}\mathbf{\Phi}_n^{k}=\mathbf{\Phi}_n^{*}=\mathbf{\Phi}^{*}, \ \forall \ n\in \mathcal{{T}}.
\end{equation}
Plugging \eqref{convergence diff} into \eqref{lambda update}, we can obtain $\mathbf{\Lambda}_n^k = -\nabla f(\mathbf{\Phi}_n^k)$. Since gradient $\nabla f_n(\mathbf{\Phi})$ is Lipschtz continuous, it means $\nabla f(\mathbf{\Phi}_n^k)$ is bounded. Therefore, we conclude that $\mathbf{\Lambda}_n^k$ is also bounded. Combining this result with \eqref{convergence Lambda}, we can see that $\mathbf{\Lambda}_n^k$ can converge to some limit point, i.e.,
\begin{equation}\label{convergence_Lambda}
  \lim\limits_{k\rightarrow+\infty}\mathbf{\Lambda}_n^{k}=\mathbf{\Lambda}_n^{*}, \ \forall \ n\in \mathcal{{T}},
\end{equation}
which finish the proof for \eqref{conv variables} in Theorem \ref{theorem admm}.

Second, we prove $\mathbf{\Phi}^{*}$ is some stationary point of problem \eqref{unconstrained model}. Since $\mathbf{\Phi}^{k+1}\!=\!\underset{0\preceq\mathbf{\Phi}\prec2\pi} {\arg \min} \ \mathcal{L}\!\left(\!\mathbf{\Phi},  \{\mathbf{\Phi}_n^k,\mathbf{\Lambda}_n^k,n\!\in\!\mathcal{T}\}\!\right)$ and function $\mathcal{L}\left(\mathbf{\Phi}, \{\mathbf{\Phi}_n^k, \mathbf{\Lambda}_n^k, n\in\mathcal{{T}}\}\right)$ is quadratic function with respect to $\mathbf{\Phi}$, we have
\begin{equation*}\label{stationary}
\left\langle\nabla_{\!\mathbf{\Phi}}\mathcal{L}(\mathbf{\Phi}^{k+1}\!,\!\{\!\mathbf{\Phi}_n^{k},\!\mathbf{\Lambda}_n^{k},n\in\mathcal{{T}}\}\!),
\mathbf{\Phi}\!-\!\mathbf{\Phi}^{k+1}\right\rangle\!\geq\!0,0\preceq\mathbf{\Phi}\prec\!2\pi,
\end{equation*}
which can be further derived as
\begin{equation}\label{stationary h_ori}
\hspace{-0.05cm}\left\langle\!\!\!-\!\!\!\sum_{n\in\mathcal{T}}\!\!\left( \rho_n(\!\mathbf{\Phi}_n^k\!-\!\mathbf{\Phi}^{k+1}\!)\!+\!\mathbf{\Lambda}_n^k\right)\!,\!\mathbf{\Phi}\!-\!\mathbf{\Phi}^{k+1}\!\!\right\rangle\!\!\geq\!0,\!0\!\preceq\!\mathbf{\Phi}\!\prec\!2\pi.
\end{equation}
When $k\rightarrow+\infty$, plugging the convergence results \eqref{convergence_Phi}-\eqref{convergence_Lambda} into \eqref{stationary h_ori}, it can be simplified as
\begin{equation}\label{stationary h}
\left\langle-\!\sum_{n\in\mathcal{T}}\mathbf{\Lambda}_n^*,\mathbf{\Phi}-\mathbf{\Phi}^{*}\right\rangle\geq0,\ \ 0\preceq\mathbf{\Phi}\prec2\pi.
\end{equation}
Since $\nabla f_n({\mathbf{\Phi}}^*)= - \mathbf{\Lambda}_n^*,\ \ \forall n\in\mathcal{{T}}$, \eqref{stationary h} can be further derived as
\begin{equation}\label{stationary phi}
\left\langle\sum_{n\in\mathcal{{T}}}\nabla f_n({\mathbf{\Phi}}^*),\mathbf{\Phi}-\mathbf{\Phi}^{*}\right\rangle\geq0,\ \ 0\preceq\mathbf{\Phi}\prec2\pi.
\end{equation}
which completes the proof. $\hfill\blacksquare$

\section{Proof of Theorem \ref{theorem pdmm}}\label{proof theorem pdmm}
Let $\mathbf\Phi^*$ be the limit point when consensus-PDMM algorithm is convergent. Then, to show $\mathbf\Phi^*$ is some stationary point of problem \eqref{unconstrained model}, we prove that it should satisfy the following inequality
\begin{equation}\label{stat point}
\left\langle \sum_{n\in \mathcal{T}}\nabla f_n({\mathbf{\Phi}^{*}}),\mathbf{\Phi}-\mathbf{\Phi}^{*}\right\rangle\geq0,0\preceq\mathbf{\Phi}\prec2\pi.
\end{equation}

Since $\mathbf{{\Phi}}^{k+1} = \underset{0\preceq\mathbf{\Phi}\prec2\pi}{\arg\min}\ \ \mathcal{U}\left(\mathbf{\Phi},\{\mathbf{\Phi}_n^k,\mathbf{\Lambda}_n^k,n\in \mathcal{T}\backslash 0\}\right)$ and function $\mathcal{U}\left(\mathbf{\Phi}, \{\mathbf{\Phi}_n^k, \mathbf{\Lambda}_n^k, n\in\mathcal{{T}}\backslash 0\}\right)$ is quadratic with respect to $\mathbf{\Phi}$, we have
\begin{equation*}
\left\langle\nabla_{\mathbf{\Phi}}\mathcal{U}(\mathbf{\Phi}^{k+1},\{\mathbf{\Phi}_n^{k},\mathbf{\Lambda}_n^{k},n\in\mathcal{{T}}\backslash 0\}),
\mathbf{\Phi}-\mathbf{\Phi}^{k+1}\right\rangle\geq0,
\end{equation*}
i.e.,
\begin{equation}\label{stationary U ori}
\begin{split}
\hspace{-0.2cm}&\bigg\langle\!\!\nabla f_0\left(\mathbf{\Phi}^k\right)\!+\!L(\mathbf{\Phi}^{k\!+\!1}\!-\!\mathbf{\Phi}^k)\\
\hspace{-0.2cm}&\ -\!\!\!\sum_{n\in\mathcal{T}\backslash 0}\!\left( \rho_n(\mathbf{\Phi}^{k\!+\!1}-\mathbf{\Phi}_n^k)\!+\!\mathbf{\Lambda}_n^k\right),\mathbf{\Phi}-\mathbf{\Phi}^{k+1}\bigg\rangle\geq0,
\end{split}
\end{equation}
where $0\preceq\mathbf{\Phi}\prec2\pi$.
Since \eqref{conv variables} holds, i.e., $\lim\limits_{k\rightarrow+\infty}\mathbf{\Phi}^{k}=\mathbf{\Phi}^{*}$, $\lim\limits_{k\rightarrow+\infty}\mathbf{\Phi}_n^{k}=\mathbf{\Phi}_n^{*}$ and $\mathbf{\Phi}^{*}=\mathbf{\Phi}_n^{*}$, \eqref{stationary U ori} can be  derived as \eqref{stationary U} when $k\rightarrow +\infty$.
\begin{equation}\label{stationary U}
\left\langle\!\nabla f_0\left(\mathbf{\Phi}^*\right)-\!\!\sum_{n\in\mathcal{T}\backslash 0}\mathbf{\Lambda}_n^*,\mathbf{\Phi}-\mathbf{\Phi}^{*}\right\rangle\!\geq\!0,\  0\preceq\mathbf{\Phi}\prec2\pi.
\end{equation}
By solving problem \eqref{step2 PDMM}, we can get
\begin{equation}\label{solve pdmm b}
\begin{split}
\nabla\! f_n({\mathbf{\Phi}}_n^{k})\!+\! \mathbf{\Lambda}_n^k\!+\!L_n(\mathbf{\Phi}_n^{k+1}\!\!-\!\mathbf{\Phi}_n^{k}) \!+ \!\rho_n(\mathbf{\Phi}_n^{k+1}\!\!-\!\mathbf{\Phi}^{k})\!=\!0.
\end{split}
\end{equation}
Since $\lim\limits_{k\rightarrow+\infty}\mathbf{\Phi}^{k}=\mathbf{\Phi}^{*}$, $\lim\limits_{k\rightarrow+\infty}\mathbf{\Phi}_n^{k}=\mathbf{\Phi}_n^{*}$ and $\mathbf{\Phi}^{*}=\mathbf{\Phi}_n^{*}$, \eqref{solve pdmm b} can be rewritten as
$\nabla f_n({\mathbf{\Phi}}^*)= - \mathbf{\Lambda}_n^*,\ \ \forall n\in\mathcal{{T}}\backslash 0$.
Plugging it into \eqref{stationary U}, we can obtain \eqref{stat point}. This completes the proof. $\hfill\blacksquare$

\section{Proof of Theorem \ref{theorem local opt}}\label{proof local opt}
First, we define the following quantities
\begin{equation}
\begin{split}
\mathbf{x}&={\rm vec}(\mathbf{X})=\left[\mathbf{x}_1;\mathbf{x}_2;\cdots;\mathbf{x}_M\right],\\
\mathbf{B}_i& = \left[\mathbf{0}_{N\times(i-1)N},\mathbf{I}_{N},\mathbf{0}_{N\times(M-i)N}\right].
\end{split}
\end{equation}
Then, function $f_n(\mathbf{x})$ in problem \eqref{unconstrained model} can be expressed as
\begin{equation}
  \begin{split}
f_n(\mathbf{x})& = \!\sum_{i=1}^M\sum_{j=1}^M\!\left|\mathbf{x}^H\mathbf{B}_i^H\mathbf{S}_n\mathbf{B}_j\mathbf{x}\right|^2\!-\!M\!N^2\delta_n\\
               & = \!\sum_{i=1}^M\sum_{j=1}^M \!\left|\!{\rm vec}(\mathbf{x}\mathbf{x}^H)^H{\rm vec}(\mathbf{B}_i^H\!\mathbf{S}_n\mathbf{B}_j)\!\right|^2\!\!-\!M\!N^2\delta_n\\
               &= \!\mathbf{y}^H\mathbf{Q}_n\mathbf{y}\!-\!M\!N^2\delta_n,
  \end{split}
\end{equation}
where $\mathbf{Q}_n=\sum\limits_{i=1}^M\sum\limits_{j=1}^M{\rm vec}(\mathbf{B}_i^H\mathbf{S}_n\mathbf{B}_j){\rm vec}(\mathbf{B}_i^H\mathbf{S}_n\mathbf{B}_j)^H$ and $\mathbf{y}\!=\!{\rm vec}(\mathbf{x}\mathbf{x}^H)$. Let $\mathbf{Q}=\sum\limits_{n\in{\mathcal{T}}}\mathbf{Q}_n$, the objective function in \eqref{unconstrained model} can be rewritten as
\begin{equation}
  \begin{split}
f(\mathbf{x})=\sum\limits_{n\in{\mathcal{T}}} f_n(\mathbf{\Phi})= \mathbf{y}^H\mathbf{Q}\mathbf{y}\!-\!M\!N^2\delta_n.
  \end{split}
\end{equation}
Since $\mathbf{x}={\rm vec}(e^{j\mathbf{\Phi}}), 0\preceq\mathbf{\Phi}\prec2\pi$ and $\mathbf{Q}$ is a Hermitian matrix, following the analysis of nonconvex quartic minimization problem in \cite{Kisialiou_09}\cite{Wang_12}, we can conclude that
 any local minima $\mathbf\Phi^*$ of problem \eqref{unconstrained model} is a $\frac{1}{2}$-approximation of its global minimum, which establishes \eqref{local opt f}. $\hfill\blacksquare$

\ifCLASSOPTIONcaptionsoff
  \newpage
\fi

\end{document}